\documentclass[11pt, letterpaper]{article}
\usepackage[margin=1in]{geometry}

\usepackage{amsfonts}
\usepackage{amsmath}
\usepackage{amssymb}
\usepackage{amsthm}
\usepackage{bbm} 
\usepackage{cite}
\usepackage{color}
\usepackage{graphicx}
\usepackage{mathrsfs} 

\def\extraspacing{\vspace{3mm} \noindent}
\def\figcapup{\vspace{-1mm}}
\def\figcapdown{\vspace{-0mm}}

\def\vgap{\vspace{1mm}}

\def\tabpos{\hspace{4mm} \= \hspace{4mm} \= \hspace{4mm} \= \hspace{4mm} \= \hspace{4mm} \= \hspace{4mm} \= \hspace{4mm} \= \hspace{4mm} \= \hspace{4mm} \= \hspace{4mm} \= \hspace{4mm} \= \hspace{4mm} \= \hspace{4mm} \= \hspace{4mm} 
\kill}
\newcommand{\mytab}[1]{\begin{tabbing}\tabpos #1\end{tabbing}}

\newcommand{\minipg}[2]{\begin{center}\begin{minipage}{#1}#2\end{minipage}\end{center}}
\newcommand{\myitems}[1]{\begin{itemize} #1 \end{itemize}}
\newcommand{\myenums}[1]{\begin{enumerate} #1 \end{enumerate}}

\newcommand{\bm}[1]{\textrm{\boldmath${#1}$}}
\newcommand{\mb}[1]{\mathbf{#1}}

\newcommand{\myeqn}[1]{\begin{eqnarray}#1\end{eqnarray}}

\newcommand{\set}[1]{\{#1\}}

\newcommand{\explain}[1]{(\textrm{#1})}

\def\mit{\mathit}

\def\eps{\epsilon}
\def\fr{\frac}
\def\-{\mbox{-}}

\def\real{\mathbb{R}}

\def\tO{\tilde{O}}

\def\lc{\lceil}

\def\rc{\rceil}

\def\nn{\nonumber}

\def\Pr{\mathbf{Pr}}

\DeclareMathOperator*{\argmax}{arg\,max}
\DeclareMathOperator*{\polylog}{polylog}

\allowdisplaybreaks

\usepackage{cite}
\usepackage{mathtools}

\theoremstyle{plain}
\newtheorem{theorem}{Theorem}[section]

\newtheorem{lemma}[theorem]{Lemma}

\theoremstyle{definition}

\theoremstyle{remark}

\usepackage{balance}
\usepackage{microtype}
\usepackage{wasysym}
\allowdisplaybreaks

\def\vgap{\vspace{1mm}}
\def\extraspacing{\vspace{3mm} \noindent}
\def\figcapup{\vspace{-3mm}}
\def\figcapdown{\vspace{-3mm}}

\def\C{\mathcal{C}}
\def\E{\mathcal{E}}
\def\G{\mathcal{G}}
\def\H{\mathcal{H}}
\def\Q{\mathcal{Q}}
\def\S{\mathcal{S}}
\def\T{\mathcal{T}}
\def\V{\mathcal{V}}
\def\X{\mathcal{X}}
\def\Y{\mathcal{Y}}
\def\Z{\mathcal{Z}}

\def\back{\mathit{back}}
\def\agm{\mathit{AGM}}
\def\att{\mathbf{att}}
\def\can{\mathit{constraint}}
\def\dc{\textsf{DC}}
\def\deg{\mathit{deg}}
\def\dir{\mathit{dir}}
\def\dom{\mathbf{dom}}
\def\In{\mathrm{IN}}
\def\join{\mathit{join}}
\def\mb{\mathit{modular}}
\def\modu{\textsf{M}}
\def\mycycle{\pentagon}
\def\mystar{\star}
\def\Out{\mathrm{OUT}}
\def\occ{\mathit{occ}}
\def\pass{\mathrm{pass}}
\def\polym{\mathit{polymat}}
\def\reld{\mathit{reldeg}}
\def\schema{\mathit{schema}}
\def\sub{\mathit{sub}}
\def\undir{\mathit{undir}}

\newcommand*{\ldb}{\{\mskip-5mu\{}
\newcommand*{\rdb}{\}\mskip-5mu\}}

\title{Join Sampling under Acyclic Degree Constraints and \\ (Cyclic) Subgraph Sampling}

\author{
    Ru Wang \hspace{20mm} Yufei Tao\\[3mm]
    Department of Computer Science and Engineering \\
    Chinese University of Hong Kong \\
    \{{\em rwang21}, {\em taoyf}\}{\em @cse.cuhk.edu.hk}
}

\begin{document}

\maketitle

\begin{abstract}
    Given a (natural) join with an acyclic set of degree constraints (the join itself does not need to be acyclic), we show how to draw a uniformly random sample from the join result in $O(\polym/ \max \set{1, \Out})$ expected time (assuming data complexity) after a preprocessing phase of $O(\In)$ expected time, where $\In$, $\Out$, and $\polym$ are the join's input size, output size, and polymatroid bound, respectively. This compares favorably with the state of the art (Deng et al.\ and Kim et al., both in PODS'23), which states that, in the absence of degree constraints, a uniformly random sample can be drawn in $\tO(\agm / \max \set{1, \Out})$ expected time after a preprocessing phase of $\tO(\In)$ expected time, where $\agm$ is the join's AGM bound and $\tO(.)$ hides a $\polylog (\In)$ factor. Our algorithm applies to every join supported by the solutions of Deng et al.\ and Kim et al. Furthermore, since the polymatroid bound is at most the AGM bound, our performance guarantees are never worse, but can be considerably better, than those of Deng et al.\ and Kim et al.

    \vgap

    We then utilize our techniques to tackle {\em directed subgraph sampling}, a problem that has extensive database applications and bears close relevance to joins. Let $G = (V, E)$ be a directed data graph where each vertex has an out-degree at most $\lambda$, and let $P$ be a directed pattern graph with a constant number of vertices. The objective is to uniformly sample an occurrence of $P$ in $G$. The problem can be modeled as join sampling with input size $\In = \Theta(|E|)$ but, whenever $P$ contains cycles, the converted join has {\em cyclic} degree constraints. We show that it is always possible to throw away certain degree constraints such that (i) the remaining constraints are acyclic and (ii) the new join has asymptotically the same polymatroid bound $\polym$ as the old one. Combining this finding with our new join sampling solution yields an algorithm to sample from the original (cyclic) join (thereby yielding a uniformly random occurrence of $P$) in $O(\polym/ \max \set{1, \Out})$ expected time after $O(|E|)$ expected-time preprocessing, where $\Out$ is the number of occurrences. We also prove similar results for {\em undirected subgraph sampling} and demonstrate how our techniques can be significantly simplified in that scenario. Previously, the state of the art for (undirected/directed) subgraph sampling uses $O(|E|^{\rho^*} / \max \set{1, \Out})$ time to draw a sample (after $O(|E|)$ expected-time preprocessing) where $\rho^*$ is the fractional edge cover number of $P$. Our results are more favorable because $\polym$ never exceeds but can be considerably lower than $|E|^{\rho^*}$.
\end{abstract}

\newpage

\section{Introduction} \label{sec:intro}

In relational database systems, (natural) joins are acknowledged as notably computation-intensive, with its cost surging drastically in response to expanding data volumes. In the current big data era, the imperative to circumvent excessive computation increasingly overshadows the requirement for complete join results. A myriad of applications, including machine learning algorithms, online analytical processing, and recommendation systems, can operate effectively with random samples. This situation has sparked research initiatives focused on devising techniques capable of producing samples from a join result significantly faster than executing the join in its entirety. In the realm of graph theory, the significance of join operations is mirrored in their intrinsic connections to {\em subgraph listing}, a classical problem that seeks to pinpoint all the occurrences of a pattern $P$ (for instance, a directed 3-vertex cycle) within a data graph $G$ (such as a social network where a directed edge symbolizes a ``follow'' relationship). Analogous to joins, subgraph listing demands a vast amount of computation time, which escalates rapidly with the sizes of $G$ and $P$. Fortunately, many social network analyses do not require the full set of occurrences of $P$, but can function well with only samples from those occurrences. This has triggered the development of methods that can extract samples considerably faster than finding all the occurrences.

\vgap

This paper will revisit {\em join sampling} and {\em subgraph sampling} under a unified ``degree-constrained framework''. Next, we will first describe the framework formally in Section~\ref{sec:intro:prob}, review the previous results in Section~\ref{sec:intro:related}, and then overview our results in  Section~\ref{sec:intro:ours}.

\subsection{Problem Definitions} \label{sec:intro:prob}

\noindent {\bf Join Sampling.} Let $\att$ be a finite set, with each element called an {\em attribute}, and $\dom$ be a countably infinite set, with each element called a {\em value}. For a non-empty set $\X \subseteq \att$ of attributes, a {\em tuple} over $X$ is a function $\bm{u}: \X \rightarrow \dom$. For any non-empty subset $\Y \subseteq \X$, we define $\bm{u}[\Y]$ --- the {\em projection} of $\bm{u}$ on $Y$ --- as the tuple $\bm{v}$ over $\Y$ satisfying $\bm{v}(Y) = \bm{u}(Y)$ for every attribute $Y \in \Y$.

\vgap

A {\em relation} $R$ is a set of tuples over the same set $\Z$ of attributes; we refer to $\Z$ as the {\em schema} of $R$ and represent it as $\schema(R)$. The {\em arity} of $R$ is the size of $\schema(R)$. For any subsets $\X$ and $\Y$ of $\schema(R)$ satisfying $\X \subset \Y$ (note: $\X$ is a {\em proper} subset of $\Y$), define:
\myeqn{
    \deg_{\Y \mid \X}(R) &=& \max_{\textrm{tuple $\bm{u}$ over $\X$}} \Big| \Big\{\bm{v}[\Y] \mid \bm{v} \in R, \bm{v}[\X] = \bm{u} \Big\} \Big|.
    \label{eqn:degree}
}
For an intuitive explanation, imagine grouping the tuples of $R$ by $\X$ and counting, for each group, how many {\em distinct} $\Y$-projections are formed by the tuples therein. Then, the value $\deg_{\Y \mid \X}(R)$ corresponds to the maximum count of all groups. It is worth pointing out that, when $\X = \emptyset$, then $\deg_{\Y \mid \X}(R)$ is simply $|\Pi_\Y (R)|$ where $\Pi$ is the standard ``projection'' operator in relational algebra. If in addition $\Y = \schema(R)$, then $\deg_{\Y \mid \X}(R)$ equals $|R|$.

\vgap

We define a {\em join} as a set $\Q$ of relations (some of which may have the same schema). Let $\schema(\Q)$ be the union of the attributes of the relations in $\Q$, i.e., $\schema(\Q) = \bigcup_{R \in Q} \schema(R)$. Focusing on ``data complexity'', we consider only joins where both $\Q$ and $\schema(\Q)$ have constant sizes. The result of $\Q$ is a relation over $\schema(\Q)$ formalized as:
\myeqn{
    \join(\Q) = \set{\textrm{tuple $\bm{u}$ over $\schema(\Q)$} \mid \forall R \in \Q: \bm{u}[\schema(R)] \in R}. \nn
}
Define $\In = \sum_{R \in \Q} |R|$ and $\Out = |\join(\Q)|$. We will refer to $\In$ and $\Out$ as the {\em input size} and {\em output size} of $Q$, respectively.

\vgap

A {\em join sampling} operation returns a tuple drawn uniformly at random from $\join(\Q)$ or declares $\join(\Q) = \emptyset$. All such operations must be mutually independent. The objective of the {\em join sampling problem} is to preprocess the input relations of $\Q$ into an appropriate data structure that can be used to perform join-sampling operations repeatedly.

\vgap

We study the problem in the scenario where $\Q$ conforms to a set $\dc$ of degree constraints. Specifically, each {\em degree constraint} has the form $(\X, \Y, N_{\Y|\X})$ where $\X$ and $\Y$ are subsets of $\schema(\Q)$ satisfying $\X \subset \Y$ and $N_{\Y|\X} \ge 1$ is an integer. A relation $R \in \Q$ is said to {\em guard} the constraint $(\X, \Y, N_{\Y|\X})$ if
\myeqn{
    \text{$\Y \subseteq \schema(R)$, and $\deg_{\Y \mid \X}(R) \le N_{\Y|\X}$}. \nn
}
The join $\Q$ is {\em consistent} with $\dc$ --- written as $\Q \models \dc$ --- if every degree constraint in $\dc$ is guarded by at least one relation in $\Q$. It is safe to assume that $\dc$ does not have two constraints $(\X, \Y, N_{\Y\mid\X})$ and $(\X', \Y', N_{\Y'\mid\X'})$ with $\X = \X'$ and $\Y = \Y'$; otherwise, assuming $N_{\Y\mid\X} \le N_{\Y'\mid\X'}$, the constraint $(\X', \Y', N_{\Y'\mid\X'})$ is redundant and can be removed from $\dc$.

\vgap

In this work, we concentrate on ``acyclic'' degree dependency. To formalize this notion, let us define a {\em constraint dependency graph} $G_\dc$ as follows. This is a directed graph whose vertex set is $\schema(\Q)$ (i.e., each vertex of $G_\dc$ is an attribute in $\schema(\Q)$). For each degree constraint $(\X, \Y, N_{\Y|\X})$ such that $\X \ne \emptyset$, we add a (directed) edge $(X, Y)$ to $G_\dc$ for every pair $(X, Y) \in \X \times (\Y - \X)$. We say that the set $\dc$ is {\em acyclic} if $G_\dc$ is an acyclic graph; otherwise, $\dc$ is {\em cyclic}.

\vgap

It is important to note that each relation $R \in \Q$ implicitly defines a special degree constraint $(\X, \Y, N_{\Y|\X})$ where $\X = \emptyset$, $\Y = \schema(R)$, and $N_{\Y|\X} = |R|$. Such a constraint --- known as a {\em cardinality constraint} --- is always assumed to be present in $\dc$. As all cardinality constraints have $\X = \emptyset$, they do not affect the construction of $G_\dc$. Consequently, if $\dc$ only contains cardinality constraints, then $G_\dc$ is empty and hence trivially acyclic. Moreover, readers should avoid the misconception that ``an acyclic $G_\dc$ implies $\Q$ being an acyclic join''; these two acyclicity notions are unrelated. While the definition of an acyclic join is not needed for our discussion, readers unfamiliar with this term may refer to \cite[Chapter 6.4]{ahv95}.

\extraspacing {\bf Directed Graph Sampling.} We are given a {\em data graph} $G = (V, E)$ and a {\em pattern graph} $P = (V_P, E_P)$, both being simple directed graphs. The pattern graph is weakly-connected\footnote{Namely, if we ignore the edge directions, then $P$ becomes a connected undirected graph.} and has a constant number of vertices. A simple directed graph $G_\sub = (V_\sub, E_\sub)$ is a {\em subgraph} of $G$ if $V_\sub \subseteq V$ and $E_\sub \subseteq E$. The subgraph $G_\sub$ is an {\em occurrence} of $P$ if they are isomorphic, namely, there is a bijection $f: V_\sub \rightarrow V_P$ such that, for any distinct vertices $u_1, u_2 \in V_\sub$, there is an edge $(u_1, u_2) \in E_\sub$ if and only if $(f(u_1), f(u_2))$ is an edge in $E_P$. We will refer to $f$ as a {\em isomorphism bijection} between $P$ and $G_\sub$.

\vgap

A {\em subgraph sampling} operation returns an occurrence of $P$ in $G$ uniformly at random or declares the absence of any occurrence. All such operations need to be mutually independent. The objective of the {\em subgraph sampling problem} is to preprocess $G$ into a data structure that can support every subgraph-sampling operation efficiently. We will study the problem under a degree constraint: every vertex in $G$ has an out-degree at most $\lambda$.

\extraspacing {\bf Undirected Graph Sampling.} The setup of this problem is the same as the previous problem except that (i) both the data graph $G = (V, E)$ and the pattern graph $P = (V_P, E_P)$ are simple undirected graphs, with $P$ being connected; (ii) a subgraph $G_\sub$ of $G$ is an {\em occurrence} of $P$ if $G_\sub$ and $P$ are isomorphic in the undirected sense: there is a {\em isomorphism bijection} $f: V_\sub \rightarrow V_P$ between $P$ and $G_\sub$ such that, for any distinct $u_1, u_2 \in V_\sub$, an edge $\{u_1, u_2\}$ exists in $E_\sub$ if and only if $\{f(u_1), f(u_2)\} \in E_P$;\footnote{We represent a directed edge as an ordered pair and an undirected edge as a set.} and (iii) the degree constraint becomes: every vertex in $G$ has a degree at most $\lambda$.

\extraspacing {\bf Math Conventions.} For an integer $x \ge 1$, the notation $[x]$ denotes the set $\{1, 2, ..., x\}$; as a special case, $[0]$ represents the empty set. Every logarithm $\log(\cdot)$ has base $2$, and function $\exp_2(x)$ is defined to be $2^x$. We use double curly braces to represent multi-sets, e.g., $\ldb 1, 1, 1, 2, 2, 3\rdb$ is a multi-set with 6 elements.

\subsection{Related Work} \label{sec:intro:related}

\noindent {\bf Join Computation.} Any algorithm correctly answering a join query $\Q$ must incur $\Omega(\Out)$ time just to output the $\Out$ tuples in $\join(\Q)$. Hence, finding the greatest possible value of $\Out$ is an imperative step towards unraveling the time complexity of join evaluation. A classical result in this regard is the {\em AGM bound} \cite{agm13}. To describe this bound, let us define the {\em schema graph} of $\Q$ as a multi-hypergraph $\G = (\V, \E)$ where
\myeqn{
    \V = \schema(\Q) \text{, and } 
    \E = \ldb \schema(R) \mid R \in \Q \rdb. \label{eqn:schema-graph-V-E}
}
Note that $\E$ is a multi-set because the relations in $\Q$ may have identical schemas. A {\em fractional edge cover} of $\G$ is a function $w: \E \rightarrow [0, 1]$ such that, for any $X \in \V$, it holds that $\sum_{F \in \E: X \in F} w(F) \ge 1$ (namely, the total weight assigned to the hyperedges covering $X$ is at least 1). Atserias, Grohe, and Marx \cite{agm13} showed that, given any fractional edge cover, it always holds that $\Out \le \prod_{F \in \E} |R_F|^{w(F)}$, where $R_F$ is the relation in $\Q$ whose schema corresponds to the hyperedge $F$. The AGM bound is defined as $\agm(\Q) = \min_{w} \prod_{F \in \E} |R_F|^{w(F)}$.

\vgap

The AGM bound is tight: given any hypergraph $\G = (\V, \E)$ and any set of positive integers $\set{N_F \mid F \in \E}$, there is always a join $\Q$ such that $\Q$ has $\G$ as the schema graph, $|R_F| = |N_F|$ for each $F \in \E$, and the output size $\Out$ is $\Theta(\agm(\Q))$. This has motivated the development of algorithms \cite{abs21,dlt23,knrr16,kns17,jr18,nprr12,nnrr14,nprr18,nrr13,nrr20,v14} that can compute $\join(\Q)$ in $\tO(\agm(\Q))$ time --- where $\tO(.)$ hides a factor polylogarithmic to the input size $\In$ of $\Q$ --- and therefore are worst-case optimal up to an $\tO(1)$ factor.

\vgap

However, the tightness of the AGM bound relies on the assumption that all the degree constraints on $\Q$ are purely cardinality constraints. In reality, general degree constraints are prevalent, and their inclusion could dramatically decrease the maximum output size $\Out$. This observation has sparked significant interest \cite{dsbc23,gt17c,jr18,knos23,kns16,kns17,n18,s23} in establishing refined upper bounds on $\Out$ tailored for more complex degree constraints. Most notably, Khamis et al.\ \cite{kns17} proposed the {\em entropic bound}, which  is applicable to any set $\dc$ of degree constraints and is tight in a strong sense (see Theorem 5.5 of \cite{s23}). Unfortunately, the entropic bound is difficult to compute because it requires solving a linear program (LP) involving infinitely many constraints (it remains an open problem whether the computation is decidable). Not coincidentally, no join algorithm is known to have a running time matching the entropic bound.

\vgap

To circumvent the above issue, Khamis et al.\ \cite{kns17} introduced the {\em polymatroid bound} as an alternative, which we represent as $\polym(\dc)$ because this bound is fully decided by $\dc$ (i.e., any join $\Q \models \dc$ must satisfy $\Out \le \polym(\dc)$). Section~\ref{sec:pre} will discuss $\polym(\dc)$ in detail; for now, it suffices to understand that (i) the polymatroid bound, although possibly looser than the entropic bound, never exceeds the AGM bound, and (ii) $\polym(\dc)$ can be computed in $O(1)$ time under data complexity. Khamis et al.\ \cite{kns17} proposed an algorithm named \textsf{PANDA} that can evaluate an arbitrary join $\Q \models \dc$ in time $\tO(\polym(\dc))$.

\vgap

Interestingly, when $\dc$ is acyclic, the entropic bound is equivalent to the polymatroid bound \cite{n18}. In this scenario, Ngo \cite{n18} presented a simple algorithm to compute any join $\Q \models \dc$ in $O(\polym(\dc))$ time, after a preprocessing of $O(\In)$ expected time.

\extraspacing {\bf Join Sampling.} For an acyclic join (not to be confused with a join having an acyclic set of degree constraints), it is possible to sample from the join result in constant time, after a preprocessing of $O(\In)$ expected time \cite{zcl+18}. The problem becomes more complex when dealing with an arbitrary (cyclic) join $\Q$, with the latest advancements presented in two PODS'23 papers \cite{dlt23,khfh23}. Specifically, Kim et al.\ \cite{khfh23} described how to sample in $\tO(\agm(\Q) / \max\{1, \Out\})$ expected time, after a preprocessing of $\tO(\In)$ time. Deng et al.\ \cite{dlt23} achieved the same guarantees using different approaches, and offered a rationale explaining why the expected sample time $O(\agm(\Q) / \Out)$ can no longer be significantly improved, even when $0 < \Out \ll \agm(\Q)$, subject to commonly accepted conjectures. We refer readers to \cite{agpr99,cmn99,cy20,dlt23,khfh23,zcl+18} and the references therein for other results (now superseded) on join sampling.


\extraspacing {\bf Subgraph Listing.} Let us start by clarifying the {\em fractional edge cover number} $\rho^*(P)$ of a simple undirected pattern graph $P = (V_P, E_P)$. Given a fractional edge cover of $P$ (i.e., function $w: E_P \rightarrow [0, 1]$ such that, for any vertex $X \in V_P$, we have $\sum_{F\in E_P: X \in F} w(F) \ge 1$), define $\sum_{F \in E_P} w(F)$ as the {\em total weight} of $w$. The value of $\rho^*(P)$ is the smallest total weight of all fractional edge covers of $P$. Given a directed pattern graph $P$, we define its fractional edge cover number $\rho^*(P)$ as the value $\rho^*(P')$ of the corresponding undirected graph $P'$ that is obtained from $P$ by ignoring all the edge directions.

\vgap

When $P$ has a constant size, it is well-known \cite{a81, agm13} that any data graph $G = (V, E)$ can encompass $O(|E|^{\rho^*(P)})$ occurrences of $P$. This holds true both when $P$ and $G$ are directed and when they are undirected. This upper bound is tight: in both the directed and undirected scenarios, for any integer $m$, there is a data graph $G = (V, E)$ with $|E| = m$ edges that has $\Omega(m^{\rho^*(P)})$ occurrences of $P$. Thus, a subgraph listing algorithm is considered worst-case optimal if it finishes in $\tO(|E|^{\rho^*(P)})$ time.

\vgap

It is well-known that directed/undirected subgraph listing can be converted to a join $\Q$ on binary relations (namely, relations of arity 2). The join $\Q$ has an input size of $\In = \Theta(|E|)$, and its AGM bound is $\agm(\Q) = \Theta(|E|^{\rho^*(P)})$. All occurrences of $P$ in $G$ can be derived from $\join(\Q)$ for free. Thus, any $\tO(\agm(\Q))$-time join algorithm is essentially worst-case optimal for subgraph listing.

\vgap

Assuming $P$ and $G$ to be directed, Jayaraman et al.\ \cite{jrr21} presented interesting enhancement over the above transformation in the scenario where each vertex of $G$ has an out-degree at most $\lambda$. The key lies in examining the polymatroid bound of the join $\Q$ derived from subgraph listing. As will be explained in Section~\ref{sec:directed-subgraph}, this join $\Q$ has a set $\dc$ of degree constraints whose constraint dependency graph $G_\dc$ coincides with $P$. Jayaraman et al.\ developed an algorithm that lists all occurrences of $\Q$ in $G$ in $O(\polym(\dc))$ time (after a preprocessing of $O(\In)$ expected time) and confirmed that this is worst-case optimal. Their findings are closely related to our work, and we will delve into them further when their specifics become crucial to our discussion.

\vgap

There is a substantial body of literature on bounding the cost of subgraph listing using parameters distinct from those already mentioned. These studies typically concentrate on specific patterns (such as paths, cycles, and cliques) or particular graphs (for instance, those that are sparse under a suitable metric). We refer interested readers to \cite{akls22,bfnn19,bpwz14,cn85,e99,hkss13,jx23,m17,np85,s81} and the references therein.

\extraspacing {\bf Subgraph Sampling.} Fichtenberger, Gao, and Peng \cite{fgp20} described how to sample an occurrence of the pattern $P$ in the data graph $G$ in $O(|E|^{\rho^*(P)} / \max\{1, \Out\})$ expected time, where $\Out$ is the number of occurrences of $P$ in $G$, after a preprocessing of $O(|E|)$ expected time. In \cite{dlt23}, Deng et al.\ clarified how to deploy an arbitrary join sampling algorithm to perform subgraph sampling; their approach ensures the same guarantees as in \cite{fgp20}, baring an $\tO(1)$ factor. The methods of \cite{fgp20, dlt23} are applicable in both undirected and directed scenarios.

\subsection{Our Results} \label{sec:intro:ours}

For any join $\Q$ with an acyclic set $\dc$ of degree constraints, we will demonstrate in Section~\ref{sec:join-sample} how to extract a uniformly random sample from $\join(\Q)$ in $O(\polym(\dc) / \max\set{1, \Out})$ expected time, following an initial preprocessing of $O(\In)$ expected time. This performance is favorable when compared to the recent results of \cite{dlt23, khfh23} (reviewed in Section~\ref{sec:intro:related}), which examined settings where $\dc$ consists only of cardinality constraints and is therefore trivially acyclic. As $\polym(\dc)$ is at most but can be substantially lower than $\agm(\Q)$, our guarantees are never worse, but can be considerably better, than those in \cite{dlt23, khfh23}.

\vgap

What if $\dc$ is cyclic? An idea, proposed in \cite{n18}, is to discard enough constraints to make the remaining set $\dc'$ of constraints acyclic (while ensuring $\Q \models \dc'$). Our algorithm can then be applied to draw a sample in $O(\polym(\dc') / \max\set{1, \Out})$ time. However, this can be unsatisfactory because $\polym(\dc')$ can potentially be much larger than $\polym(\dc)$.

\vgap

Our next contribution is to prove that, interestingly, the issue does not affect subgraph listing/sampling. Consider first directed subgraph listing, defined by a pattern graph $P$ and a data graph $G$ where every vertex has an out-degree at most $\lambda$. This problem can be converted to a join $\Q$ on binary relations, which is associated with a set $\dc$ of degree constraints such that the constraint dependency graph $G_\dc$ is exactly $P$. Consequently, whenever $P$ contains a cycle, so does $G_\dc$, making $\dc$ cyclic. Nevertheless, we will demonstrate in Section~\ref{sec:directed-subgraph} the existence of an acyclic set $\dc' \subset \dc$ ensuring $Q \models \dc'$ and $\polym(\dc) = \Theta(\polym(\dc'))$. This ``magical'' $\dc'$ has an immediate implication: Ngo's join algorithm in \cite{n18}, when applied to $Q$ and $\dc'$ directly, already solves directed subgraph listing optimally in $O(\polym(\dc'))$ $=$ $O(\polym(\dc))$ time. This dramatically simplifies --- in terms of both procedure and analysis --- an algorithm of Jayaraman et al.\ \cite{jrr21} (for directed subgraph listing, reviewed in Section~\ref{sec:intro:related}) that has the same guarantees.

\vgap

The same elegance extends to directed subgraph sampling: by applying our new join sampling algorithm to $\Q$ and the ``magical'' $\dc'$, we can sample an occurrence of $P$ in $G$ using $O(\polym(\dc) / \max\{1, \Out\})$ expected time, after a preprocessing of $O(|E|)$ expected time. As $\polym(\dc)$ never exceeds but can be much lower than $\agm(\Q) = \Theta(|E|^{\rho^*(P)})$, our result compares favorably with the state of the art \cite{dlt23, khfh23, fgp20} reviewed in Section~\ref{sec:intro:related}.

\vgap

Undirected subgraph sampling (where both $P$ and $G$ are undirected and each vertex in $G$ has a degree at most $\lambda$) is a special version of its directed counterpart and can be settled by a slightly modified version of our directed subgraph sampling algorithm. However, it is possible to do better by harnessing the undirected nature. In Section~\ref{sec:undirected-subgraph}, we will first solve the polymatroid bound into a {\em closed-form} expression, which somewhat unexpectedly exhibits a crucial relationship to a well-known graph decomposition method. This relationship motivates a surprisingly simple algorithm for undirected subgraph sampling that offers guarantees analogous to those in the directed scenario.

\vgap

By virtue of the power of sampling, our findings have further implications on other fundamental problems including output-size estimation, output permutation, and small-delay enumeration. We will elaborate on the details in Section~\ref{sec:conclusion}.

\section{Preliminaries} \label{sec:pre}

\noindent {\bf Set Functions, Polymatroid Bounds, and Modular Bounds.} Suppose that $\S$ is a finite set. We refer to a function $h: 2^\S \rightarrow \real_{\ge 0}$ as a {\em set function over $\S$}, where $\real_{\ge 0}$ is the set of non-negative real values. Such a function $h$ is said to be
\myitems{
    \item {\em zero-grounded} if $h(\emptyset) = 0$;
    \item {\em monotone} if $h(\X) \le h(\Y)$ for all $\X, \Y$ satisfying $\X \subseteq \Y \subseteq \S$;
    \item {\em modular} if $h(\X) = \sum_{A \in \X} h(\{A\})$ holds for any $\X \subseteq \S$; 
    \item {\em submodular} if $h(\X \cup \Y) + h(\X \cap \Y) \le h(\X) + h(\Y)$ holds for any $\X, \Y \subseteq \S$.
}
Define:
\myeqn{
    \modu_\S &=& \text{the set of modular set functions over $\S$} \nn \\ 
    \Gamma_\S &=& \text{the set of set functions over $\S$ that are zero-grounded, monotone, submodular} \nn 
}
Note that every modular function must be zero-grounded and monotone. Clearly, $\modu_\S \subseteq \Gamma_\S$.

\vgap

Consider $\C$ to be a set of triples, each having the form $(\X, \Y, N_{\Y\mid\X})$ where $X \subset \Y \subseteq \S$ and $N_{\Y\mid\X}\ge 1$ is an integer. We will refer to $\C$ as a {\em rule collection} over $\S$ and to each triple therein as a {\em rule}. Intuitively, the presence of a rule collection is to instruct us to focus only on certain restricted set functions. Formally, these are the set functions in:
\myeqn{
    \H_\C = \big\{ \text{set function $h$ over $\S$} \mid h(\Y) - h(\X) \le \log N_{\Y \mid \X}, \hspace{2mm} \forall (\X, \Y, N_{\Y \mid \X}) \in \C
    \big\}. \label{eqn:H-C}
}
The {\em polymatroid bound} of $\C$ can now be defined as
\myeqn{
    \polym(\C) &=& \exp_2 \Big(\max_{h \in \Gamma_\S \cap \H_\C} h(\S) \Big). \label{eqn:polym-C}
}
Recall that $\exp_2(x) = 2^x$. Similarly, the {\em modular bound} of $\C$ is defined as
\myeqn{
    \mb(\C) &=& \exp_2 \Big(\max_{h \in \modu_\S \cap \H_\C} h(\S) \Big). \label{eqn:modular-C}
}

\extraspacing {\bf Join Output Size Bounds.} Let us fix a join $\Q$ whose schema graph is $\G = (\V, \E)$. Suppose that $\Q$ is consistent with a set $\dc$ of degree constraints, i.e., $\Q \models \dc$. As explained in Section~\ref{sec:intro:prob}, we follow the convention that each relation of $\Q$ implicitly inserts a cardinality constraint (i.e., a special degree constraint) to $\dc$. Note that the set $\dc$ is merely a rule collection over $\V$. The following lemma was established by Khamis et al.\ \cite{kns17}:

\begin{lemma} [\cite{kns17}] \label{lmm:pre-out-atmost-polym}
    The output size $\Out$ of $\Q$ is at most {\em $\polym(\dc)$}, i.e., the polymatroid bound of {\em $\dc$} (as defined in \eqref{eqn:polym-C}).
\end{lemma}

How about $\mb(\dc)$, i.e., the modular bound of $\V$? As $\modu_\V \subseteq \Gamma_\V$, we have $\mb(\dc) \le \polym(\dc)$ and the inequality can be strict in general. However, an exception arises when $\dc$ is acyclic, as proved in \cite{n18}:

\begin{lemma} [{\cite{n18}}] \label{lmm:pre-modular-equals-polym}
    When {\em $\dc$} is acyclic, it always holds that {\em $\mb(\dc) = \polym(\dc)$}, namely, {\em $\max_{h \in \Gamma_\V \cap \H_\dc} h(\V) = \max_{h \in \modu_\V \cap \H_\dc} h(\V)$}.
\end{lemma}

As a corollary, when $\dc$ is acyclic, the value of $\mb(\dc)$ always serves as an upper bound of $\Out$. In our technical development, we will need to analyze the set functions $h^* \in \Gamma_\V$ that realize the polymatriod bound, i.e., $h^*(\V) = \max_{h \in \Gamma_\V \cap \H_\dc} h(\V)$. A crucial advantage provided by Lemma~\ref{lmm:pre-modular-equals-polym} is that we can instead scrutinize those set functions $h^* \in \modu_\V$ realizing the {\em modular} bound, i.e., $h^*(\V) = \max_{h \in \modu_\V \cap \H_\dc} h(\V)$. Compared to their submodular counterparts, modular set functions exhibit more regularity because every $h \in \modu_\V$ is fully determined by its value $h(\set{A})$ on each {\em individual} attribute $A \in \V$. In particular, for any $h \in \modu_\V \cap \H_\dc$, it holds true that $h(\Y) - h(\X) = \sum_{A \in \Y - \X} h(A)$ for any $\X \subset \Y \subseteq \V$. If we associate each $A \in \V$ with a variable $\nu_A$, then $\max_{h \in \modu_\V \cap \H_\dc} h(\V)$ --- hence, also $\max_{h \in \Gamma_\V \cap \H_\dc} h(\V)$ --- is precisely the optimal value of the following LP:

\mytab{
    \> {\bf modular LP} \>\>\>\>\> max $\sum_{A \in \V} \nu_A$ subject to \\[2mm]
    \>\>\>\> $\sum\limits_{A \in \Y - \X} \nu_A \le \log N_{\Y \mid \X}$ \>\>\>\>\>\>\>\> $\forall (\X, \Y, N_{\Y\mid\X}) \in \dc$ \\
    \>\>\>\> $\nu_A \ge 0$ \>\>\>\>\>\>\>\> $\forall A \in \V$
}

\noindent We will also need to work with the LP's dual. Specifically, if we associate a variable $\delta_{\Y\mid\X}$ for every degree constraint $(\X, \Y, N_{\Y\mid\X}) \in \dc$, then the dual LP is:

\mytab{
    \> {\bf dual modular LP} \>\>\>\>\>\>\> min $\sum\limits_{(\X, \Y, N_{\Y\mid\X}) \in \dc} \delta_{\Y\mid\X} \cdot \log N_{\Y\mid\X}$ subject to \\[2mm]
    \>\>\>\> $\sum\limits_{\substack{(\X, \Y, N_{\Y\mid\X}) \in \dc \\ A \in \Y - \X}} \delta_{\Y\mid\X} \ge 1$ \>\>\>\>\>\>\>\> $\forall A \in \V$ \\
    \>\>\>\> $\delta_{\Y\mid\X} \ge 0$ \>\>\>\>\>\>\>\> $\forall (\X, \Y, N_{\Y\mid\X}) \in \dc$
}

\section{Join Sampling under Acyclic Degree Dependency} \label{sec:join-sample}

This section serves as a proof of our first main result:

\begin{theorem} \label{thm:samp}
    For any join $\Q$ consistent with an acyclic set {\em $\dc$} of degree constraints, we can build in $O(\In)$ expected time a data structure that supports each join sampling operation in {\em $O(\polym(\dc) / \max\{1, \Out\})$} expected time, where $\In$ and $\Out$ are the input and out sizes of $\Q$, respectively, and {\em $\polym(\dc)$} is the polymatroid bound of {\em $\dc$}.
\end{theorem}

\noindent {\bf Basic Definitions.} Let $\G = (\V, \E)$ be the schema graph of $\Q$, and $G_\dc$ be the constraint dependency graph determined by $\dc$. For each hyperedge $F \in \E$, we denote by $R_F$ the relation whose schema corresponds to $F$. Recall that every constraint $(\X, \Y, N_{\Y \mid \X}) \in \dc$ is guarded by at least one relation in $\Q$. Among them, we arbitrarily designate one relation as the constraint's {\em main guard}, whose schema is represented as $F(\X, \Y)$ (the main guard can then be conveniently identified as $R_{F(\X,\Y)}$).

\vgap

Set $k = |\V|$. As $G_\dc$ is a DAG (acyclic directed graph), we can order its $k$ vertices (i.e., attributes) into a topological order: $A_1, A_2, ..., A_k$. For each $i \in [k]$, define $\V_i = \set{A_1, A_2, ..., A_i}$; specially, define $\V_0 = \emptyset$. For any $i \in [k]$, define
\myeqn{
    \dc(A_i) &=& \{ (\X, \Y, N_{\Y \mid \X}) \in \dc \mid A_i \in \Y - \X
               \}
               \label{eqn:dc-A}
}

Fix an arbitrary $i \in [k]$ and an arbitrary constraint $(\X, \Y, N_{\Y \mid \X}) \in \dc(A_i)$. Given a tuple $\bm{w}$ over $\V_{i-1}$ (note: if $i = 1$, then $\V_{i-1} = \emptyset$ and $\bm{w}$ is a null tuple) and a value $a \in \dom$, we define a ``relative degree'' for $a$ as:
\myeqn{
    \reld_{i, \X, \Y} (\bm{w}, a) &=&
    \fr{
        \big|
            \sigma_{A_i = a} (\Pi_\Y (R_{F(\X, \Y)} \ltimes \bm{w}))
        \big|
    }
    {
        \big|
            \Pi_\Y (R_{F(\X, \Y)} \ltimes \bm{w})
        \big|
    }
    \label{eqn:reldeg}
}
where $\sigma$ and $\ltimes$ are the standard selection and semi-join operators in relational algebra, respectively. To understand the intuition behind $\reld_{i, \X, \Y} (\bm{w}, a)$, imagine drawing a tuple $\bm{u}$ from $\Pi_\Y (R_{F(\X, \Y)} \ltimes \bm{w})$ uniformly at random; then $\reld_{i, \X, \Y} (\bm{w}, a)$ is the probability to see $\bm{u}(A_i) = a$. Given a tuple $\bm{w}$ over $\V_{i-1}$ and a value $a \in \dom$, define
\myeqn{
    \reld^*_{i} (\bm{w}, a) &=&
    \max_{(\X, \Y, N_{\Y \mid \X}) \in \dc(A_i)} \reld_{i, \X, \Y} (\bm{w}, a)
    \label{eqn:reldeg-star} \\
    \can^*_i(\bm{w}, a) &=&
    \argmax_{(\X, \Y, N_{\Y \mid \X}) \in \dc(A_i)} \reld_{i, \X, \Y} (\bm{w}, a).
    \label{eqn:canonical-constr}
}
Specifically, $\can^*_i(\bm{w}, a)$ returns the constraint $(\X, \Y, N_{\Y \mid \X}) \in \dc(A_i)$ satisfying the condition $\reld_{i, \X, \Y} (\bm{w}, a) = \reld^*_{i} (\bm{w}, a)$. If more than one constraint meets this condition, define $\can^*_i(\bm{w}, a)$ to be an arbitrary one among those constraints.

\vgap

Henceforth, we will fix an arbitrary optimal solution $\set{ \delta_{\Y\mid \X} \mid (\X, \Y, N_{\Y \mid \X}) \in \dc }$ to the dual modular LP in Section~\ref{sec:pre}. Thus:
\myeqn{
    \prod_{(\X, \Y, N_{\Y \mid \X}) \in \dc} N_{\Y \mid \X}^{\delta_{\Y \mid \X}}
    &=&
    \exp_2 \Big(\sum\limits_{(\X, \Y, N_{\Y\mid\X}) \in \dc} \delta_{\Y\mid\X} \cdot \log N_{\Y\mid\X} \Big)
    =
    \exp_2 \Big(\max_{h \in \modu_\V \cap \H_\dc} h(\V) \Big) \nn \\
    \explain{by \eqref{eqn:modular-C}}
    &=&
    \mb(\dc) \nn \\
    \explain{by Lemma~\ref{lmm:pre-modular-equals-polym}}
    &=& \polym(\dc).
    \label{eqn:join-sample:help1}
}
Finally, for any $i \in [0, k]$ and any tuple $\bm{w}$ over $\V_i$, define:
\myeqn{
    B_i(\bm{w})
    &=&
    \prod_{(\X, \Y, N_{\Y\mid\X}) \in \dc}
        \big(
            \deg_{\Y \mid \X}(R_{F(\X, \Y)} \ltimes \bm{w})
        \big)^{\delta_{\Y \mid \X}}. \label{eqn:B_i}
}
Two observations will be useful later:
\myitems
{
    \item If $i = 0$, then $\bm{w}$ is a null tuple and $B_0(\text{null}) = \prod_{(\X, \Y, N_{\Y\mid\X}) \in \dc} (\deg_{\Y \mid \X}(R_{F(\X, \Y)}))^{\delta_{\Y\mid\X}}$, which is at most $\prod_{(\X, \Y, N_{\Y\mid\X}) \in \dc} N_{\Y \mid \X}^{\delta_{\Y\mid\X}} = \polym(\dc)$.

    \item If $i = k$ and $\bm{w} \in \join(\Q)$, then $R_{F(\X, \Y)} \ltimes \bm{w}$ contains exactly one tuple for any $(\X, \Y, N_{\Y\mid\X}) \in \dc$ and thus $B_k(\bm{w}) = 1$.
}

\begin{figure}
\mytab{
    \textsf{ADC-sample} \\
    0. \> $A_1, A_2, ..., A_k \leftarrow$ a topological order of $G_\dc$ \\
    1. \> $\bm{w}_0 \leftarrow$ a null tuple \\
    2. \> {\bf for} $i = 1$ to $k$ {\bf do} \\
    3. \>\> pick a constraint $(\X^\circ, \Y^\circ, N_{\Y^\circ \mid \X^\circ})$ uniformly at random from $\dc(A_{i})$ \\
    4. \>\> $\bm{u}^\circ \leftarrow$ a tuple chosen uniformly at random from $\Pi_{\Y^\circ} (R_{F(\X^\circ, \Y^\circ)} \ltimes \bm{w}_{i-1})$ \\
    \>\> /* note: if $i=1$, then $R_{F(\X^\circ, \Y^\circ)} \ltimes \bm{w}_{i-1} = R_{F(\X^\circ, \Y^\circ)}$ */ \\
    5. \>\> $a_i \leftarrow \bm{u}^\circ(A_i)$ \\
    6. \>\> {\bf if} $(\X^\circ, \Y^\circ, N_{\Y^\circ \mid \X^\circ}) \ne \can^*_{i-1}(\bm{w}_{i-1}, a_i)$ {\bf then} declare {\bf failure} \\
    7.\>\> $\bm{w}_i \leftarrow$ the tuple over $\V_i$ formed by concatenating $\bm{w}_{i-1}$ with $a_i$ \\
    8.\>\> declare {\bf failure} with probability $1 - p_\pass (i, \bm{w}_{i-1}, \bm{w}_i)$, where $p_\pass$ is given in \eqref{eqn:p-pass} \\
    9.\> {\bf if} $\bm{w}_k[F] \in R_F$ for $\forall F \in \E$ {\bf then} \hspace{10mm} /* that is, $\bm{w}_k \in \join(\Q)$ */\\
    10.\>\> {\bf return} $\bm{w}_k$
}
\figcapup
\caption{Our sampling algorithm}
\label{alg:join-samp}
\figcapdown
\end{figure}

\noindent {\bf Algorithm.} Our sampling algorithm, named \textsf{ADC-sample}, is presented in Figure~\ref{alg:join-samp}. At a high level, it processes one attribute at a time according to the topological order $A_1, A_2, ..., A_k$. The for-loop in Lines 2-9 finds a value $a_i$ for attribute $A_i$ ($i \in [k]$). The algorithm may fail to return anything, but when it {\em succeeds} (i.e., reaching Line 10), the values $a_1, a_2, ..., a_k$ will make a uniformly random tuple from $\join(\Q)$.

\vgap

Next, we explain the details of the for-loop. The loop starts with values $a_1, a_2, ..., a_{i-1}$ already stored in a tuple $\bm{w}_{i-1}$ (i.e., $\bm{w}_{i-1}(A_j) = a_j$ for all $j \in [i-1]$). Line 3 randomly chooses a degree constraint $(\X^\circ, \Y^\circ, N_{\Y^\circ \mid \X^\circ})$ from $\dc(A_{i})$; see \eqref{eqn:dc-A}. Conceptually, next we identify the main guard $R_{F(\X^\circ,\Y^\circ)}$ of this constraint, semi-join the relation with $\bm{w}_{i-1}$, and project the semi-join result on $\Y^\circ$ to obtain $\Pi_{\Y^\circ} (R_{F(\X^\circ, \Y^\circ)} \ltimes \bm{w}_{i-1})$. Then, Line 4 randomly chooses a tuple $\bm{u}^\circ$ from $\Pi_{\Y^\circ} (R_{F(\X^\circ, \Y^\circ)} \ltimes \bm{w}_{i-1})$ and Line 5 takes $\bm{u}^\circ(A_i)$ as the value of $a_i$ (note: $A_i \in \Y^\circ - \X^\circ \subseteq \Y^\circ$). Physically, however, we do not compute $\Pi_{\Y^\circ} (R_{F(\X^\circ, \Y^\circ)} \ltimes \bm{w}_{i-1})$ during the sample process. Instead, with proper preprocessing (discussed later), we can acquire the value $a_i$ in $O(1)$ time. Continuing, Line 6 may declare failure and terminate \textsf{ADC-sample}, but if we get past this line, $(\X^\circ, \Y^\circ, N_{\Y^\circ \mid \X^\circ})$ must be exactly $\can^*_i(\bm{w}_{i-1}, a_i)$; see \eqref{eqn:canonical-constr}. As clarified later, the check at Line 6 can be performed in $O(1)$ time. We now form a tuple $\bm{w}_i$ that takes value $a_j$ on attribute $A_j$ for each $j \in [i]$ (Line 7). Line 8 allows us to pass with probability
\myeqn{
    p_\pass(i, \bm{w}_{i-1}, \bm{w}_i)
    &=&
    \fr{B_i(\bm{w}_i)}{B_{i-1} (\bm{w}_{i-1})} \cdot
    \fr{1}{\reld^*_i(\bm{w}_{i-1}, \bm{w}_{i}(A_i))}
    \label{eqn:p-pass}
}
or otherwise terminate the algorithm by declaring failure. As proved later, $p_\pass(i, \bm{w}_{i-1}, \bm{w}_i)$ cannot exceed 1 (Lemma~\ref{lmm:p-pass}); moreover, this value can be computed in $O(1)$ time. The overall execution time of \textsf{ADC-sample} is constant.



\extraspacing {\bf Analysis.} Next we prove that the value in \eqref{eqn:p-pass} serves as a legal probability value.

\begin{lemma} \label{lmm:p-pass}
    For every $i \in [k]$, we have $p_\pass(i, \bm{w}_{i-1}, \bm{w}_i) \le 1$.
\end{lemma}

\begin{proof}
    Consider an arbitrary constraint $(\X, \Y, N_{\Y\mid\X}) \in \dc(A_i)$. Recall that \textsf{ADC-sample} processes the attributes by the topological order $A_1, ..., A_k$. In the constrained dependency graph $G_\dc$, every attribute of $\X$ has an out-going edge to $A_i$. Hence, all the attributes in $\X$ must be processed prior to $A_i$. This implies that all the tuples in $R_{F(\X, \Y)} \ltimes \bm{w}_{i-1}$ must have the same projection on $\X$. Therefore, $\deg_{\Y \mid \X}(R_{F(\X, \Y)} \ltimes \bm{w}_{i-1})$ equals $|\Pi_{\Y} (R_{F(\X,\Y)} \ltimes \bm{w}_{i-1})|$. By the same reasoning, $\deg_{\Y \mid \X}(R_{F(\X, \Y)} \ltimes \bm{w}_{i})$ equals $|\Pi_{\Y} (R_{F(\X,\Y)} \ltimes \bm{w}_{i})|$. We thus have:
    \myeqn{
        \fr{\deg_{\Y \mid \X}(R_{F(\X, \Y)} \ltimes \bm{w}_{i})}{\deg_{\Y \mid \X}(R_{F(\X, \Y)} \ltimes \bm{w}_{i-1})}
        &=&
        \fr{|\Pi_{\Y} (R_{F(\X,\Y)} \ltimes \bm{w}_{i})|}{|\Pi_{\Y} (R_{F(\X,\Y)} \ltimes \bm{w}_{i-1})|} \nn \\
        &=&
        \fr{
            \big|
                \sigma_{A_i = a_i} (\Pi_\Y (R_{F(\X, \Y)} \ltimes \bm{w}_{i-1}))
            \big|
        }
        {
            \big|
                \Pi_\Y (R_{F(\X, \Y)} \ltimes \bm{w}_{i-1})
            \big|
        } \nn \\
        &=&
        \reld_{i,\X,\Y}(\bm{w}_{i-1}, a_i) \nn \\
        &\le&
        \reld^*_{i}(\bm{w}_{i-1}, a_i).
        \label{eqn:proof:p-pass:help1}
    }
    On the other hand, for any constraint $(\X, \Y, N_{\Y\mid\X}) \notin \dc(A_i)$, it trivially holds that
    \myeqn{
        \deg_{\Y \mid \X}(R_{F(\X, \Y)} \ltimes \bm{w}_{i})
        &\le&
        \deg_{\Y \mid \X}(R_{F(\X, \Y)} \ltimes \bm{w}_{i-1})
        \label{eqn:proof:p-pass:help2}
    }
    because $R_{F(\X, \Y)} \ltimes \bm{w}_{i}$ is a subset of $R_{F(\X, \Y)} \ltimes \bm{w}_{i-1}$.

    \vgap

    We can now derive
    \myeqn{
        p_\pass(i, \bm{w}_{i-1}, \bm{w}_i) 
        &=&
        \fr{1}
        {\reld^*_i(\bm{w}_{i-1}, a_i)}
        \prod_{(\X, \Y, N_{\Y\mid\X}) \in \dc}
        \Big(
        \fr{
            \deg_{\Y \mid \X}(R_{F(\X, \Y)} \ltimes \bm{w}_i)
        }
        {
            \deg_{\Y \mid \X}(R_{F(\X, \Y)} \ltimes \bm{w}_{i-1})
        }
        \Big)^{\delta_{\Y \mid \X}} \nn \\
        \explain{by \eqref{eqn:proof:p-pass:help2}}
        &\le&
        \fr{1}
        {\reld^*_i(\bm{w}_{i-1}, a_i)}
        \prod_{(\X, \Y, N_{\Y\mid\X}) \in \dc(A_i)}
        \Big(
        \fr{
            \deg_{\Y \mid \X}(R_{F(\X, \Y)} \ltimes \bm{w}_i)
        }
        {
            \deg_{\Y \mid \X}(R_{F(\X, \Y)} \ltimes \bm{w}_{i-1})
        }
        \Big)^{\delta_{\Y \mid \X}}
        \nn \\
        \explain{by \eqref{eqn:proof:p-pass:help1}}
        &\le&
        \fr{1}
        {\reld^*_i(\bm{w}_{i-1}, a_i)}
        \prod_{(\X, \Y, N_{\Y\mid\X}) \in \dc(A_i)}
        \reld^*_i(\bm{w}_{i-1}, a_i)^{\delta_{\Y \mid \X}}
        \nn \\
        &=&
        \reld^*_i(\bm{w}_{i-1}, a_i)^{\big(\sum_{(\X, \Y, N_{\Y\mid\X}) \in \dc(A_i)} \delta_{\Y\mid\X} \big) - 1}
        \le
        1. \nn
    }
    The last step used $\sum_{(\X, \Y, N_{\Y\mid\X}) \in \dc(A_i)} \delta_{\Y\mid\X} \ge 1$ guaranteed by the dual modular LP.
\end{proof}

Next, we argue that every result tuple $\bm{v} \in \join(\Q)$ is returned by \textsf{ADC-sample} with the same probability. For this purpose, let us define two random events for each $i \in [k]$:
\myitems{
    \item event {\bf E1}$(i)$: $(\X^\circ, \Y^\circ, N_{\Y^\circ \mid \X^\circ}) = \can^*_i(\bm{w}_{i-1}, \bm{v}(A_i))$ in the $i$-th loop of \textsf{ADC-sample};
    \item event {\bf E2}$(i)$: Line 8 does not declare failure in the $i$-th loop of \textsf{ADC-sample}.
}
The probability for \textsf{ADC-sample} to return $\bm{v}$ can be derived as follows.
\myeqn{
    \Pr[\bm{v} \text{ returned}] 
    &=&
    \prod_{i=1}^k \Pr[a_i = \bm{v}(A_i), \text{\bf E1}(i), \text{\bf E2}(i) \mid \bm{w}_{i-1} = \bm{v}[\V_{i-1}]] \nn \\
    && \explain{if $i = 1$, then $\bm{w}_{i-1} = \bm{v}[\V_{i-1}]$ becomes $\bm{w}_0 = \bm{v}[\emptyset]$, which is vacuously true} \nn \\
    &=&
    \prod_{i=1}^k \Big( \Pr[a_i = \bm{v}(A_i), \text{\bf E1}(i) \mid \bm{w}_{i-1} = \bm{v}[\V_{i-1}]] \cdot \nn \\[-2mm]
    && \hspace{5mm} \Pr[\text{\bf E2}(i) \mid \text{\bf E1}(i), a_i = \bm{v}(A_i), \bm{w}_{i-1} = \bm{v}[\V_{i-1}]]
    \Big).
    \label{eqn:sample:analysis:help1}
}
Observe
\myeqn{
    && \Pr[
        a_i = \bm{v}(A_i), \text{\bf E1}(i) \mid \bm{w}_{i-1} = \bm{v}[\V_{i-1}]
    ] \nn \\
    &=&
    \Pr[
        \text{\bf E1}(i) \mid \bm{w}_{i-1} = \bm{v}[\V_{i-1}]
    ]
    \cdot
    \Pr[
        a_i = \bm{v}(A_i) \mid \text{\bf E1}(i), \bm{w}_{i-1} = \bm{v}[\V_{i-1}]
    ]
    \nn \\
    &=&
    \fr{1}{|\dc(A_i)|}
    \cdot
    \fr{
        \big|
            \sigma_{A_i = \bm{v}(A_i)} (\Pi_\Y (R_{F(\X^\circ, \Y^\circ)} \ltimes \bm{v}[\V_{i-1}]))
        \big|
    }
    {
        \big|
            \Pi_\Y (R_{F(\X^\circ, \Y^\circ)} \ltimes \bm{v}[\V_{i-1}])
        \big|
    } \nn \\
    &&
    \explain{note: $(\X^\circ, \Y^\circ, N_{\Y^\circ \mid \X^\circ}) = \can^*_i(\bm{v}[\V_{i-1}], \bm{v}(A_i))$, due to {\bf E1}$(i)$ and $\bm{w}_{i-1} = \bm{v}[\V_{i-1}]]$} \nn \\
    &=&
    \fr{1}{|\dc(A_i)|} \cdot \reld_{i,\X^\circ,\Y^\circ}(\bm{v}[\V_{i-1}], \bm{v}(A_i)) 
    =
    \fr{1}{|\dc(A_i)|} \cdot \reld^*_{i}(\bm{v}[\V_{i-1}], \bm{v}(A_i)).
    \label{eqn:sample:analysis:help2}
}
On the other hand:
\myeqn{
    && \Pr[\text{\bf E2}(i) \mid \text{\bf E1}(i), a_i = \bm{v}(A_i), \bm{w}_{i-1} = \bm{v}[\V_{i-1}]]
    \nn \\
    &=&
    p_\pass(i, \bm{v}[\V_{i-1}], \bm{v}[\V_{i}]) \nn \\
    \explain{by \eqref{eqn:p-pass}}
    &=&
    \fr{B_i(\bm{v}[\V_i])}{B_{i-1}(\bm{v}[\V_{i-1}])}
    \cdot
    \fr{1}{\reld^*_i(\bm{v}[\V_{i-1}], \bm{v}(A_i))}. \label{eqn:sample:analysis:help3}
}
Plugging \eqref{eqn:sample:analysis:help2} and \eqref{eqn:sample:analysis:help3} into \eqref{eqn:sample:analysis:help1} yields
\myeqn{
    \Pr[\bm{v} \text{ returned}]
    &=&
    \prod_{i=1}^k \fr{B_i(\bm{v}[\V_i])}{B_{i-1}(\bm{v}[\V_{i-1}])} \cdot \fr{1}{|\dc(A_i)|}
    =
    \fr{B_k(\bm{v}[\V_k])}{B_0(\bm{v}[\V_0])} \cdot \prod_{i=1}^k \fr{1}{|\dc(A_i)|} \nn \\
    &=&
    \fr{1}{B_0(\text{null})} \cdot \prod_{i=1}^k \fr{1}{|\dc(A_i)|}.
    \nn
}
As the above is identical for every $\bm{v} \in \join(\Q)$, we can conclude that each tuple in the join result gets returned by \textsf{ADC-sample} with the same probability. As an immediate corollary, each run of \textsf{ADC-sample} successfully returns a sample from $\join(\Q)$ with probability
\myeqn{
    \fr{\Out}{B_0(\text{null})} \cdot \prod_{i=1}^k \fr{1}{|\dc(A_i)|}
    \ge
    \fr{\Out}{\polym(\dc)} \cdot \prod_{i=1}^k \fr{1}{|\dc(A_i)|}
    =
    \Omega\Big(\fr{\Out}{\polym(\dc)}\Big). \nn
}
In Appendix~\ref{app:preprocessing}, we will explain how to preprocess the relations of $\Q$ in $O(\In)$ expected time to ensure that \textsf{ADC-sample} completes in $O(1)$ time.

\extraspacing {\bf Performing a Join Sampling Operation.} Recall that this operation must either return a uniformly random sample of $\join(\Q)$ or declare $\join(\Q) = \emptyset$. To support this operation, we execute two {\em threads} concurrently. The first thread repeatedly invokes \textsf{ADC-sample} until it successfully returns a sample. The other thread runs Ngo's algorithm in \cite{n18} to compute $\join(\Q)$ {\em in full}, after which we can declare $\join(\Q) \ne \emptyset$ or sample from $\join(\Q)$ in constant time. As soon as one thread finishes, we manually terminate the other one.

\vgap

This strategy guarantees that the join operation completes in $O(\polym(\dc) / \max \set{1, \Out})$ time. To explain why, consider first the scenario where $\Out \ge 1$. In this case, we expect to find a sample with $O(\polym(\dc) / \Out)$ repeats of \textsf{ADC-sample}. Hence, the first thread finishes in $O(\polym(\dc) / \Out)$ expected sample time. On the other hand, if $\Out = 0$, the second thread will finish in $O(\polym(\dc))$ time. This concludes the proof of Theorem~\ref{thm:samp}.

\extraspacing {\bf Remarks.} When $\dc$ has only cardinality constraints (is thus ``trivially'' acyclic), \textsf{ADC-sample} simplifies into the sampling algorithm of Kim et al.\ \cite{khfh23}. In retrospect, two main obstacles prevent an obvious extension of their algorithm to an arbitrary acyclic $\dc$. The first is identifying an appropriate way to deal with constraints $(\X, \Y, N_{\Y \mid \X}) \in \dc$ where $\X \ne \emptyset$ (such constraints are absent in the degenerated context of \cite{khfh23}). The second obstacle involves determining how to benefit from a topological order (attribute ordering is irrelevant in \cite{khfh23}); replacing the order with a non-topological one may ruin the correctness of \textsf{ADC-sample}.


\section{Directed Subgraph Sampling} \label{sec:directed-subgraph}

Given a directed pattern graph $P = (V_P, E_P)$ and a directed data graph $G = (V, E)$, we use $\occ(G, P)$ to represent the set of occurrences of $P$ in $G$. Every vertex in $G$ has an out-degree at most $\lambda$. Our goal is to design an algorithm to sample from $\occ(G, P)$ efficiently.

\vgap

Let us formulate the ``polymatroid bound'' for this problem. Given an integer $m \ge 1$, an integer $\lambda \in [1, m]$, and a pattern $P = (V_P, E_P)$, first build a rule collection $\C$ over $V_P$ as follows: for each edge $(X, Y) \in E_P$, add to $\C$ two rules: $(\emptyset, \set{X, Y}, m)$ and $(\set{X},$ $\set{X, Y},$ $\lambda)$. Then, the {\em directed polymatriod bound} of $m$, $\lambda$, and $P$ can be defined as
\myeqn{
    \polym_\dir(m, \lambda, P)
    &=&
    \polym(\C)
    \label{eqn:polym-graph-directed}
}
where $\polym(\C)$ follows the definition in \eqref{eqn:polym-C}.

\vgap

This formulation reflects how directed subgraph listing can be processed as a join. Consider a {\em companion join} $\Q$ constructed from $G$ and $P$ as follows. The schema graph of $\Q$, denoted as $\G = (\V, \E)$, is exactly $P = (V_P, E_P)$ (i.e., $\V = V_P$ and $\E = E_P$). For every edge $F = (X, Y) \in E_P$, create a relation $R_F \in \Q$ by inserting, for each edge $(x, y)$ in the data graph $G$, a tuple $\bm{u}$ with $\bm{u}(X) = x$ and $\bm{u}(Y) = y$ into $R_F$. The rule collection $\C$ can now be regarded as a set $\dc$ of degree constraints with which $\Q$ is consistent, i.e., $\Q \models \dc = \C$. The constraint dependence graph $G_\dc$ is precisely $P$. It is immediate that $\polym_\dir(|E|, \lambda, P) = \polym(\dc)$. To find all the occurrences in $\occ(G, P)$, it suffices to compute $\join(\Q)$. Specifically, every tuple $\bm{u} \in \join(\Q)$ that uses a distinct value on every attribute in $\V$ ($= V_P$) matches a unique occurrence in $\occ(G, P)$. Conversely, every occurrence in $\occ(G, P)$ matches the same number $c$ of tuples in $\join(\Q)$, where $c \ge 1$ is a constant equal to the number of automorphisms of $P$. If we denote $\Out = |\occ(G, P)|$ and $\Out_\Q = |\join(\Q)|$, it follows that $c \cdot \Out \le \Out_\Q \le \polym(\dc) = \polym_\dir(|E|, \lambda, P)$.

\vgap

The above observation suggests how directed subgraph sampling can be reduced to join sampling. First, sample a tuple $\bm{u}$ from $\join(\Q)$ uniformly at random. Then, check whether $\bm{u}(A) = \bm{u}(A')$ for any two distinct attributes $A, A' \in \V$. If so, declare failure; otherwise, declare success and return the unique occurrence matching the tuple $\bm{u}$. The success probability equals $c \cdot \Out / \Out_\Q$. In a success event, every occurrence in $\occ(G, P)$ has the same probability to be returned.

\vgap

When $P$ is acyclic, so is $G_\dc$, and thus our algorithm in Theorem~\ref{thm:samp} can be readily applied to handle a subgraph sampling operation. To analyze the performance, consider first $\Out \ge 1$. We expect to draw $O(\Out_\Q / \Out)$ samples from $\join(\Q)$ until a success event. As Theorem~\ref{thm:samp} guarantees retrieving a sample from $\join(\Q)$ in $O(\polym(\dc) / \Out_\Q)$ expected time, overall we expect to sample an occurrence from $\occ(G, P)$ in
\myeqn{
    O \Big(\fr{\polym(\dc)}{\Out_\Q} \cdot \fr{\Out_\Q}{\Out} \Big) =  O\Big(\fr{\polym(\dc)}{\Out}\Big) \nn
}
time. To prepare for the possibility of $\Out = 0$, we apply the ``two-thread approach'' in Section~\ref{sec:join-sample}. We run a concurrent thread that executes Ngo's algorithm in \cite{n18}, which finds the whole $\join(\Q)$, and hence $\occ(G, P)$, in $O(\polym(\dc))$ time, after which we can declare $\occ(G, P) = \emptyset$ or sample from $\occ(G, P)$ in $O(1)$ time. By accepting whichever thread finishes earlier, we ensure that the operation completes in $O(\polym(\dc) / \max \set{1, \Out})$ time.

\vgap

{\em The main challenge arises when $P$ is cyclic.} In this case, $G_\dc$ (which equals $P$) is cyclic. Thus, $\dc$ becomes a cyclic set of degree constraints, rendering neither Theorem~\ref{thm:samp} nor Ngo's algorithm in \cite{n18} applicable. We overcome this challenge with the lemma below.

\begin{lemma} \label{lmm:directed-throwaway-constraints}
    If {\em $\dc$} is cyclic, we can always find an acyclic subset {\em $\dc' \subset \dc$} satisfying {\em $\polym(\dc') = \Theta(\polym(\dc))$}.
\end{lemma}

The proof is presented in Appendix~\ref{app:proof:directed-throwaway-constraints}. Because $\Q \models \dc$ and $\dc'$ is a subset of $\dc$, we know that $\Q$ must be consistent with $\dc'$ as well, i.e., $\Q \models \dc'$. Therefore, our Theorem~\ref{thm:samp} can now be used to extract a sample from $\join(\Q)$ in $O(\polym_\dir(\dc') / \max\set{1, \Out_\Q})$ time. Importantly, Lemma~\ref{lmm:directed-throwaway-constraints} also permits us to directly apply Ngo's algorithm in \cite{n18} to compute $\join(\Q)$ in $O(\polym(\dc'))$ time. Therefore, we can now apply the two-thread technique to sample from $\occ(G, P)$ in
\myeqn{
    O \Big(\fr{\polym(\dc')}{\max\{1, \Out\}} \Big)
    =
    O \Big(\fr{\polym(\dc)}{\max\{1, \Out\}} \Big)
    =
    O \Big(\fr{\polym_\dir(|E|, \lambda, P)}{\max\{1, \Out\}} \Big) \nn
}
time. We thus have arrived yet:

\begin{theorem} \label{thm:directed}
    Let $G = (V, E)$ be a simple directed data graph, where each vertex has an out-degree at most $\lambda$. Let $P = (V_P, E_P)$ be a simple weakly-connected directed pattern graph with a constant number of vertices. We can build in $O(|E|)$ expected time a data structure that supports each subgraph sampling operation in $O(\polym_\dir(|E|, \lambda, P) / \max\{1, \Out\})$ expected time, where $\Out$ is the number of occurrences of $P$ in $G$, and $\polym_\dir(|E|, \lambda, P)$ is the directed polymatrioid bound in \eqref{eqn:polym-graph-directed}.
\end{theorem}

\noindent {\bf Remarks.} For subgraph {\em listing}, Jayaraman et al.\ \cite{jrr21} presented a sophisticated method that also enables the application of Ngo's algorithm in \cite{n18} to a cyclic $P$. Given the companion join $\Q$, they employ the {\em degree uniformization technique} \cite{jr18} to generate $t = O(\polylog |E|)$ new joins $\Q_1, \Q_2, ..., \Q_t$ such that $\join(\Q) = \bigcup_{i=1}^t \join(\Q_i)$. For each $i \in [t]$, they construct an acyclic set $\dc_i$ of degree constraints (which is not always a subset of $\dc$) with the property $\sum_{i=1}^t \polym(\dc_i) \le \polym(\dc)$. Each join $\Q_i$ ($i \in [t]$) can then be processed by Ngo's algorithm in $O(\polym(\dc_i))$ time, thus giving an algorithm for computing $\join(\Q)$ (and hence $\occ(G,P)$) in $O(\polym(\dc))$ time. On the other hand, Lemma~\ref{lmm:directed-throwaway-constraints} facilitates a {\em direct} application of Ngo's algorithm to $\Q$, implying the non-necessity of degree uniformization in subgraph listing. We believe that this simplification is noteworthy and merits its own dedicated exposition, considering the critical nature of the subgraph listing problem. In the absence of Lemma~\ref{lmm:directed-throwaway-constraints}, integrating our join-sampling algorithm with the methodology of \cite{jrr21} for the purpose of subgraph sampling would require substantially more effort. Our proof of Lemma~\ref{lmm:directed-throwaway-constraints} {\em does} draw upon the analysis of \cite{jrr21}, as discussed in depth in Appendix~\ref{app:proof:directed-throwaway-constraints}.



\section{Undirected Subgraph Sampling} \label{sec:undirected-subgraph}

Given an undirected pattern graph $P = (V_P, E_P)$ and an undirected data graph $G = (V, E)$, we use $\occ(G, P)$ to represent the set of occurrences of $P$ in $G$. Every vertex in $G$ has a degree at most $\lambda$. Our goal is to design an algorithm to sample from $\occ(G, P)$ efficiently. 

\vgap

We formulate the ``polymatroid bound'' of this problem through a reduction to its directed counterpart. For $P = (V_P, E_P)$, create a  {\em directed} pattern $P' = (V_P', E_P')$ as follows. First, set $V_P' = V_P$. Second, for every edge $\{X, Y\} \in E_P$, add to $E_P$ two directed edges $(X, Y)$ and $(Y, X)$. Now, given an integer $m \ge 1$, an integer $\lambda \in [1, m]$, and an undirected pattern $P = (V_P, E_P)$, the {\em undirected polymatroid bound} of $m$, $\lambda$, and $P$ is defined as
\myeqn{
    \polym_\undir(m, \lambda, P) &=& \polym_\dir(m, \lambda, P') \label{eqn:polym-graph-undirected}
}
where the function $\polym_\dir$ is defined in \eqref{eqn:polym-graph-directed}.

\vgap

Our formulation highlights how undirected subgraph sampling can be reduced to the directed version. From $G = (V, E)$, construct a {\em directed} graph $G' = (V', E')$ by setting $V' = V$, and for every edge $\{x, y\} \in E$, adding to $E'$ two directed edges $(x, y)$ and $(y, x)$. Every occurrence in $\occ(G,P)$ matches the same number (a constant) of occurrences in $\occ(G', P')$. By resorting to Theorem~\ref{thm:directed}, we can obtain an algorithm to sample from $\occ(G,P)$ in $O(\polym_\undir(2|E|, \lambda, P) / \max\{1, \Out\})$ time (where $\Out = |\occ(G, P)|$) after a preprocessing of $O(|E|)$ expected time. We omit the details because they will be superseded by another simpler approach to be described later.

\vgap

Unlike $\polym_\dir(m,\lambda,P)$, which is defined through an LP, we can solve the undirected counterpart $\polym_\undir(m,\lambda,P)$ into a closed form. It is always possible \cite{s03,nprr18,akk19} to decompose $P$ into {\em vertex-disjoint} subgraphs
$\mycycle_1, \mycycle_2, ..., \mycycle_\alpha$, $\mystar_1, \mystar_2, ...,$ and $\mystar_\beta$
(for some integers $\alpha, \beta \ge 0$) such that
\myitems{
    \item $\mycycle_i$ is an odd-length cycle for each $i \in [\alpha]$;
    \item $\mystar_j$ is a star\footnote{A {\em star} is an undirected graph where one vertex, called the {\em center}, has an edge to every other vertex, and each non-center vertex has degree 1.} for each $j \in [\beta]$;
    \item $\sum_{i=1}^\alpha \rho^*(\mycycle_i) + \sum_{j=1}^\beta \rho^*(\mystar_j) = \rho^*(P)$; see Section~\ref{sec:intro:related} for the definition of the fractional edge cover number function $\rho^*(.)$.
}
We will refer to $(\mycycle_1, ..., \mycycle_\alpha$, $\mystar_1, ...,$  $\mystar_\beta)$ as a {\em fractional edge-cover decomposition} of $P$. Define: \vspace{-1mm}
\myeqn{
    k_\mit{cycle} &=& \text{total number of vertices in $\mycycle_1, ..., \mycycle_\alpha$} \label{eqn:k_cycle} \\
    k_\mit{star} &=& \text{total number of vertices in $\mystar_1, ..., \mystar_\beta$} \label{eqn:k_star}.
}
We establish the lemma below in Appendix~\ref{app:proof:polym-undirected-closed}. \vspace{-1mm}

\begin{lemma} \label{lmm:polym-undirected-closed}
    If $k$ is the number of vertices in $P$, then \vspace{-2mm}
    \myeqn{
        \polym_\undir(m, \lambda, P) =
        \begin{cases}
            m \cdot \lambda^{k-2} &  \text{if } \lambda \le \sqrt{m}, \\
            m^{\frac{k_\mit{cycle}}{2}+\beta} \cdot \lambda^{k_\mit{star} - 2\beta}
            &  \text{if } \lambda > \sqrt{m}
        \end{cases}
        \label{eqn:polym-undirected-closed}
    }
\end{lemma}

As shown in Appendix~\ref{app:tightness-undirected}, for $k = O(1)$, the expression in \eqref{eqn:polym-undirected-closed} asymptotically matches an upper bound of $|\occ(G', P)|$ --- for any $G'$ with $m$ edges and maximum vertex-degree at most $\lambda$ --- easily derived from a fractional edge-cover decomposition. The same appendix will also prove that, for a wide range of $m$ and $\lambda$ values, there exists a graph $G'$ with $m$ edges and maximum vertex degree at most $\lambda$ guaranteeing the following {\em simultaneously} for all patterns $P$ with $k$ vertices: $|\occ(G',P)|$ is no less than the expression's value, up to a factor of $1/(4k)^k$. This implies that, when $k$ is a constant, the expression {\em asymptotically} equals $\polym_\undir(m, \lambda, P)$. More effort is required to prove that the expression equals $\polym_\undir(m, \lambda, P))$ {\em precisely}, as we demonstrate in Appendix~\ref{app:proof:polym-undirected-closed}.

\vgap

Lemma~\ref{lmm:polym-undirected-closed} points to an alternative strategy to perform undirected subgraph sampling without joins. Identify an arbitrary spanning tree $\T$ of the pattern $P = (V_P, E_P)$. Order the vertices of $V_P$ as $A_1, A_2, ..., A_k$ such that, for any $i \in [2, k]$, vertex $A_i$ is adjacent in $\T$ to a (unique) vertex $A_j$ with $j \in [i-1]$. Now construct a map $f: V_P \rightarrow V$ as follows (recall that $V$ is the vertex set of the data graph $G = (V, E)$). First, take an edge $\set{u, v}$ uniformly at random from $E$, and choose one of the following with an equal probability: (i) $f(A_1) = u, f(A_2) = v$, or (ii) $f(A_1) = v, f(A_2) = u$. Then, for every $i \in [3, k]$, decide $f(A_i)$ as follows. Suppose that $A_i$ is adjacent to $A_j$ for some (unique) $j \in [1, i-1]$. Toss a coin with probability $\deg(f(A_j))/\lambda$, where $\deg(f(A_j))$ is the degree of vertex $f(A_j)$ in $G$. If the coin comes up tails, declare failure and terminate. Otherwise, set $f(A_i)$ to a neighbor of $f(A_j)$ in $G$ picked uniformly at random. After finalizing the map $f(.)$, check whether $\{f(A_i) \mid i \in [k]\}$ induces a subgraph of $G$ isomorphic to $P$. If so, {\em accept} $f$ and return this subgraph; otherwise, {\em reject} $f$. To guarantee a sample or declare $\occ(G, P) = \emptyset$, apply the ``two-thread approach'' by (i) repeating the algorithm until acceptance and (ii) concurrently running an algorithm for computing the whole $\occ(G, P)$ in $O(\polym_\undir(|E|, \lambda, P))$ time\footnote{For this purpose, use Ngo's algorithm in \cite{n18} to find all occurrences in $\occ(G',P')$ --- see our earlier definitions of $G'$ and $P'$ --- which is possible due to Lemma~\ref{lmm:directed-throwaway-constraints}.}. As proved in Appendix~\ref{app:undirected-sampling}, this ensures the expected sample time $O(|E| \cdot \lambda^{k-2} / \max\{1, \Out\})$ after a preprocessing of $O(|E|)$ expected time.

\vgap

The above algorithm suffices for the case $\lambda \le \sqrt{|E|}$. Consider now the case $\lambda > \sqrt{|E|}$. We  will construct a map $f: V_P \rightarrow V$ according to a given fractional edge-cover decomposition $(\mycycle_1, ..., \mycycle_\alpha$, $\mystar_1, ...,$  $\mystar_\beta)$. For each $i \in [\alpha]$, use the algorithm of \cite{fgp20} to uniformly sample an occurrence of $\mycycle_i$ --- denoted as $G_\sub(\mycycle_i)$ --- in $O(|E|^{\rho^*(\mycycle_i)} / \max\{1, |\occ(G, \mycycle_i)| \})$ expected time. Let $f_{\mycycle_i}$ be a map from the vertex set of $\mycycle_i$ to that of $G_\sub(\mycycle_i)$,  chosen uniformly at random from all the isomorphism bijections (defined in Section~\ref{sec:intro:prob}) between $\mycycle_i$ and $G_\sub(\mycycle_i)$. For each $j \in [\beta]$, apply our earlier algorithm to uniformly sample an occurrence of $\mystar_j$ --- denoted as $G_\sub(\mystar_j)$ --- in $O(|E| \cdot \lambda^{k_j} / \max\{1, |\occ(G, \mystar_j)| \})$ expected time, where $k_j$ is the number of vertices in $\mystar_j$. Let $f_{\mystar_j}$ be a map from the vertex set of $\mystar_j$ to that of $G_\sub(\mystar_j)$, chosen uniformly at random from all the polymorphism bijections between $\mystar_j$ and $G_\sub(\mystar_j)$. If any of $\mycycle_1, ..., \mycycle_\alpha$, $\mystar_1, ...,$  $\mystar_\beta$ has no occurrences, declare $\occ(G, P) = \emptyset$ and terminate. Otherwise, the functions in $\set{f_{\mycycle_i} \mid i \in [\alpha]}$ and $\set{f_{\mystar_j} \mid j \in [\beta]}$ together determine the map $f$ we aim to build. Check whether $\{f(A) \mid A \in V\}$ induces a subgraph of $G$ isomorphic to $P$. If so, {\em accept} $f$ and return this subgraph; otherwise, {\em reject} $f$. Repeat until acceptance and concurrently run an algorithm for computing the whole $\occ(G, P)$ in $O(\polym_\undir(|E|, \lambda, P))$ time. As proved in Appendix~\ref{app:undirected-sampling}, this ensures the expected sample time $O(|E|^{\frac{k_\mit{cycle}}{2}+\beta} \cdot \lambda^{k_\mit{star} - 2\beta} / \max\{1, \Out\})$ after a preprocessing of $O(|E|)$ expected time.

\vgap

We can now conclude with our last main result:

\begin{theorem} \label{thm:undirected}
    Let $G = (V, E)$ be a simple undirected data graph, where each vertex has a degree at most $\lambda$. Let $P = (V_P, E_P)$ be a simple connected pattern graph with a constant number of vertices. We can build in $O(|E|)$ expected time a data structure that supports each subgraph sampling operation in $O(\polym_\undir(|E|, \lambda, P) / \max\{1, \Out\})$ expected time, where $\Out$ is the number of occurrences of $P$ in $G$, and $\polym_\undir(|E|, \lambda, P)$ is the undirected polymatrioid bound in \eqref{eqn:polym-undirected-closed}.
\end{theorem}

\section{Concluding Remarks} \label{sec:conclusion}

Our new sampling algorithms imply new results on several other fundamental problems. We will illustrate this with respect to evaluating a join $\Q$ consistent with an acyclic set $\dc$ of degree constraints. Similar implications also apply to subgraph sampling.
\myitems{
    \item By standard techniques \cite{cy20,dlt23}, we can estimate the output size $\Out$ up to a relative error $\eps$ with high probability (i.e., at least $1-1/\In^c$ for an arbitrarily large constant $c$) in time $\tO(\fr{1}{\eps^2} \fr{\polym(\dc)}{\max\{1, \Out\}})$ after a preprocessing of $O(\In)$ expected time.

    \item Employing a technique in \cite{dlt23}, we can, with high probability, report all the tuples in $\join(\Q)$ with a delay of $\tO(\fr{\polym(\dc)}{\max\{1, \Out\}})$. In this context, `delay' refers to the maximum interval between the reporting of two successive result tuples, assuming the presence of a placeholder tuple at the beginning and another at the end.

    \item In addition to the delay guarantee, our algorithm in the second bullet can, with high probability, report the tuples of $\join(\Q)$ in a random permutation. This means that each of the $\Out!$ possible permutations has an equal probability of being the output.
}
All of the results presented above compare favorably with the current state of the art as presented in \cite{dlt23}. This is primarily due to the superiority of $\polym(\dc)$ over $\agm(\Q)$. In addition, our findings in the last two bullet points also complement Ngo's algorithm as described in \cite{n18} in a satisfying manner.

\appendix

\section*{Appendix}


\section{Implementing \textsf{ADC-Sample} with Indexes} \label{app:preprocessing}

We preprocess each constraint $(\X, \Y, N_{\Y\mid\X}) \in \dc$ as follows. Let $R \in Q$ be its main guard, i.e., $R = R_{F(\X, \Y)}$.
For each $i \in [k]$ and each tuple $\bm{w} \in \Pi_{\schema(R) \cap \V_i} (R)$, define
\myeqn{
    R_\Y(i, \bm{w}) &=& \{ \bm{u}[\Y] \mid \bm{u} \in R, \bm{u}[\schema(R) \cap \V_i] = \bm{w} \} \nn
}
which we refer to as a {\em fragment}.

\vgap

During preprocessing, we compute and store $R_\Y(i, \bm{w})$ for every $i \in [k]$ and $\bm{w} \in \Pi_{\schema(R) \cap \V_i} (R)$. Next, we will explain how to do so for an arbitrary $i \in [k]$. First, group all the tuples of $R$ by the attributes in $\schema(R) \cap \V_i$, which can be done in $O(\In)$ expected time by hashing. Then, perform the following steps for each group in turn. Let $\bm{w}$ be the group's projection on $\schema(R) \cap \V_i$. We compute the group tuples' projections onto $\Y$ and eliminate duplicate projections, the outcome of which is precisely $R_\Y(i, \bm{w})$ and is stored using an array. With hashing, this requires expected time only linear to the group's size. Therefore, the total cost of generating the fragments $R_\Y(i, \bm{w})$ of all $\bm{w} \in \Pi_{\schema(R) \cap \V_i} (R)$ is $O(\In)$ expected.

\vgap

After the above preprocessing, given any $i \in [k]$, constraint $(\X, \Y, N_{\Y\mid\X}) \in \dc$, tuple $\bm{w}$ over $\V_{i-1}$, and value $a \in \dom$, we can compute $\reld_{i,\X,\Y}(\bm{w}, a)$ defined in \eqref{eqn:reldeg} in constant time. For convenience, let $R = R_{F(\X,\Y)}$. To compute $|\Pi_\Y (R \ltimes \bm{w})|$ (the denominator of \eqref{eqn:reldeg}), first obtain $\bm{w}_1 = \bm{w}[\schema(R) \cap \V_{i-1}]$. Then, $\Pi_\Y (R \ltimes \bm{w})$ is just the fragment $R_\Y (i-1, \bm{w}_1)$, which has been pre-stored. The size of this fragment can be retrieved using $\bm{w}_1$ in $O(1)$ time. Similarly, to compute $|\sigma_{A_i=a}(\Pi_\Y (R \ltimes \bm{w}))|$ (the numerator of \eqref{eqn:reldeg}), we can first obtain $\bm{w}_2$, which is a tuple over $\schema(R) \cap \V_i$ that shares the values of $\bm{w}_1$ on all the attributes in $\schema(R) \cap \V_{i-1}$ and additionally uses value $a$ on attribute $A_i$. Then, $\sigma_{A_i=a}(\Pi_\Y (R \ltimes \bm{w}))$ is just the fragment $R_\Y(i, \bm{w}_2)$, which has been pre-stored. The size of this fragment can be fetched using $\bm{w}_2$ in $O(1)$ time.

\vgap

As a corollary, given any $i \in [k]$, tuple $\bm{w}$ over $\V_{i-1}$, and value $a \in \dom$, we can compute $\reld^*_i(\bm{w}, a)$ and $\can^*_i(\bm{w}, a)$ --- defined in \eqref{eqn:reldeg-star} and \eqref{eqn:canonical-constr}, respectively --- in constant time.

\vgap

It remains to explain how to implement Line 4 of \textsf{ADC-sample} (Figure~\ref{alg:join-samp}). Here, we want to randomly sample a tuple from $\Pi_{\Y^\circ} (R_{F(\X^\circ, \Y^\circ)} \ltimes \bm{w}_{i-1})$. Again, for convenience, let $R = R_{F(\X^\circ, \Y^\circ)}$. Obtain $\bm{w}' = \bm{w}_{i-1}[\schema(R) \cap \V_{i-1}]$. Then, $\Pi_\Y (R \ltimes \bm{w}_{i-1})$ is just the fragment $R_\Y (i-1, \bm{w}')$, which has been stored in an array. The starting address and size of the array can be acquired using $\bm{w}'$ in $O(1)$ time, after which a sample can be drawn from the fragment in constant time.

\section{Proof of Lemma \ref{lmm:directed-throwaway-constraints} } \label{app:proof:directed-throwaway-constraints}

Let us rephrase the problem as follows. Let $P = (V_P, E_P)$ be a {\em cyclic} pattern graph. Given an integer $m \ge 1$ and an integer $\lambda \in [1, m]$, define $\dc$ to be a set of degree constraints over $V_P$ that contains two constraints for each edge $(X, Y) \in E_P$: $(\emptyset, \set{X, Y}, m)$ and $(\set{X},$ $\set{X, Y},$ $\lambda)$. The constraint dependence graph $G_\dc$ is exactly $P$ and, hence, is cyclic. We want to prove the existence of an acyclic $\dc' \subset \dc$ such that $\polym(\dc') = \polym(\dc)$. We will first tackle the situation where $\lambda > \sqrt{m}$ before proceeding to the opposite scenario. The former case presents a more intriguing line of argumentation than the latter.

\extraspacing {\bf Case $\bm{\lambda > \sqrt{m}}$.} For every edge $(X, Y) \in G_\dc = (V_P, E_P)$, define two variables: $x_{X, Y}$ and $z_{X, Y}$. Jayaraman et al.\ \cite{jrr21} showed that, for $\lambda > \sqrt{m}$, $\polym(\dc)$ is, up to a constant factor, the optimal value of the following LP (named LP$^{(+)}$ following \cite{jrr21}):
\mytab{
    \> {\bf LP$\bm{^{(+)}}$} \cite{jrr21} \>\>\>\> min $\sum\limits_{(X, Y) \in E_P} x_{X, Y} \log m + z_{X,Y} \log \lambda$ subject to \\[2mm]
    \>\>\>\> $\sum\limits_{(X, A) \in E_P} (x_{X, A} + z_{X, A}) + \sum\limits_{(A, Y) \in E_P} x_{A, Y} \ge 1$ \hspace{10mm} $\forall A \in V_P$ \\
    \>\>\>\> $x_{X, Y} \ge 0$, $z_{X, Y} \ge 0$  \hspace{10mm} $\forall (X, Y) \in E_P$
}

\begin{lemma} \label{lmm:app:throwaway:opt-z-acyclic}
    There exists an optimal solution to LP$^{(+)}$ satisfying the condition that the edges in $\{ (X, Y) \in E_P \mid z_{X,Y} > 0\}$ induce an acyclic subgraph of {\em $G_\dc$}.
\end{lemma}

We note that while the above lemma is not expressly stated in \cite{jrr21}, it can be extrapolated from the analysis presented in Section H.2 of \cite{jrr21}. Nevertheless, the argument laid out in \cite{jrr21} is quite intricate. Our proof, which will be presented below, incorporates news ideas beyond their argument and is considerably shorter. Specifically, these new ideas are evidenced in the way we formulate a novel LP optimal solution in \eqref{eqn:app:throwaway-1}-\eqref{eqn:app:throwaway-4}.

\begin{proof}[Proof of Lemma~\ref{lmm:app:throwaway:opt-z-acyclic}]
    Consider an arbitrary optimal solution to LP$^{(+)}$ that sets $x_{X,Y} = x^*_{X,Y}$ and $z_{X,Y} = z^*_{X,Y}$ for each $(X,Y) \in E_P$. If the edge set $\{ (X, Y) \in E_P \mid z^*_{X,Y} > 0\}$ induces an acyclic graph, we are done. Next, we consider that $G_\dc$ contains a cycle.

    \vgap

    Suppose that $(A_1, A_2)$ is the edge in the cycle with the smallest $z^*_{A_1, A_2}$ (breaking ties arbitrarily). Let $(A_2, A_3)$ be the edge succeeding $(A_1, A_2)$ in the cycle. It thus follows that $z^*_{A_2, A_3} \ge z^*_{A_1, A_2}$. Define
    \myeqn{
        x'_{A_2, A_3} &=& x^*_{A_2, A_3} + z^*_{A_1, A_2} \label{eqn:app:throwaway-1} \\
        x'_{A_1, A_2} &=& x^*_{A_1, A_2} \label{eqn:app:throwaway-2} \\
        z'_{A_2, A_3} &=& 0 \label{eqn:app:throwaway-3} \\
        z'_{A_1, A_2} &=& 0 \label{eqn:app:throwaway-4}
    }
    For every edge $(X, Y) \in E_P \setminus \{(A_1, A_2), (A_2, A_3)\}$, set $x'_{X,Y} = x^*_{X,Y}$ and $z'_{X,Y} = z^*_{X,Y}$. It is easy to verify that, for every vertex $A \in V_P$, we have
    \myeqn{
        \sum\limits_{(X, A) \in E_P} (x'_{X, A} + z'_{X, A}) + \sum\limits_{(A, Y) \in E_P} x'_{A, Y} \ge
        \sum\limits_{(X, A) \in E_P} (x^*_{X, A} + z^*_{X, A}) + \sum\limits_{(A, Y) \in E_P} x^*_{A, Y}. \nn
    }
    Therefore, $\{x'_{X, Y}, z'_{X, Y} \mid (X, Y) \in E_P\}$ serves as a feasible solution to LP$^{(+)}$. However:
    \myeqn{
        && \Big( \sum\limits_{(X, Y) \in E_P} x'_{X, Y} \log m + z'_{X,Y} \log \lambda \Big)
        -
        \Big( \sum\limits_{(X, Y) \in E_P} x^*_{X, Y} \log m + z^*_{X,Y} \log \lambda \Big) \nn \\
        &=& z^*_{A_1, A_2} \log m - (z^*_{A_1, A_2} + z^*_{A_2, A_3}) \log \lambda \nn \\
        &\le& z^*_{A_1, A_2} \log m - 2 \cdot z^*_{A_1, A_2} \log \lambda \nn \\
        &<& 0
    }
    where the last step used the fact $\lambda^2 > m$. This contradicts the optimality of $\{x^*_{X, Y}, z^*_{X, Y} \mid (X, Y) \in E_P\}$.
\end{proof}

We now build a set $\dc'$ of degree constraints as follows. First, take an optimal solution $\set{x^*_{X, Y}, z^*_{X, Y} \mid (X, Y) \in E_P}$ to LP$^{(+)}$ promised by Lemma~\ref{lmm:app:throwaway:opt-z-acyclic}. Add to $\dc'$ a constraint $(X, \set{X, Y}, \lambda)$ for every $(X, Y) \in E_P$ satisfying $z^*_{X,Y} > 0$. Then, for every edge $(X, Y) \in E_P$, add to $\dc'$ a constraint $(\emptyset, \set{X, Y}, m)$. The $\dc'$ thus constructed must be acyclic. Denote by $G_{\dc'} = (V_P', E_P')$ the degree constraint graph of $\dc'$. Note that $V_P = V_P'$ and $E_P' \subset E_P$.

\begin{lemma} \label{lmm:app:throwaway:help1}
    The {\em $\dc'$} constructed in the above manner satisfies {\em $\polym(\dc') = \Theta(\polym(\dc))$.}
\end{lemma}

\begin{proof}
    We will first establish $\polym(\dc') \ge \polym(\dc)$. Remember that $\polym(\dc')$ is the optimal value of the modular LP (in its primal form) defined by $\dc'$, as described in Section~\ref{sec:pre}. Similarly, $\polym(\dc)$ is the optimal value of the modular LP defined by $\dc$. Given that $\dc' \subset \dc$, the LP defined by $\dc'$ incorporates only a subset of the constraints found in the LP defined by $\dc$. Therefore, it must be the case that $\polym(\dc') \ge \polym(\dc)$.

    \vgap

    The rest of the proof will show $\polym(\dc') = O(\polym(\dc))$, which will establish the lemma. Consider the following LP:
    \mytab{
        \> {\bf LP$\bm{^{(+)}_1}$} \>\>\>\> min $\sum\limits_{(X, Y) \in E_P} x_{X, Y} \log m + z_{X,Y} \log \lambda$ subject to \\[2mm]
        \>\>\>\> $\sum\limits_{(X, A) \in E_P} x_{X, A} + \sum\limits_{(A, Y) \in E_P} x_{A, Y} + \sum\limits_{(X,A) \in E'_P} z_{X, A} \ge 1$ \hspace{10mm} $\forall A \in \V_P$ \\
        \>\>\>\> $x_{X, Y} \ge 0$, $z_{X, Y} \ge 0$  \hspace{10mm} $\forall (X, Y) \in E_P$
    }
    The condition $(X,A) \in E'_P$ in the first inequality marks the difference between LP$^{(+)}_1$ and LP$^{(+)}$. Note that the two LPs have the same objective function.
    \minipg{0.9\linewidth}{
        {\bf Claim 1:} LP$^{(+)}_1$ and LP$^{(+)}$ have the same optimal value.
    }
    To prove the claim, first observe that any feasible solution $\set{x_{X, Y}, z_{X, Y} \mid (X, Y) \in E_P}$ to LP$^{(+)}_1$ is also a feasible solution to LP$^{(+)}$. Hence, the optimal value of LP$^{(+)}$ cannot exceed that of LP$^{(+)}_1$. On the other hand, recall that earlier we have identified an optimal solution $\set{x^*_{X, Y},$ $z^*_{X, Y} \mid (X, Y) \in E_P}$ to LP$^{(+)}$. By how $\dc'$ is built from that solution and how $G_{\dc'} = (V'_P, E'_P)$ is built from $\dc'$, it must hold that $z^*_{X, Y} = 0$ for every $(X, Y)$ $\in$ $E_P \setminus E'_P$. Hence, $\set{x^*_{X, Y}, z^*_{X, Y} \mid (X, Y) \in E_P}$ makes a feasible solution to LP$^{(+)}_1$. This implies that $\set{x^*_{X, Y}, z^*_{X, Y} \mid (X, Y) \in E_P}$ must be an optimal solution to LP$^{(+)}_1$. Claim 1 now follows.

    \vgap

    Consider another LP:
    \mytab{
        \> {\bf LP$\bm{^{(+)}_2}$} \>\>\>\> min $\sum\limits_{(X, Y) \in E_P} x_{X, Y} \log m + \sum\limits_{(X, Y) \in E_P'} z_{X,Y} \log \lambda$ subject to \\[2mm]
        \>\>\>\> $\sum\limits_{(X, A) \in E_P} x_{X, A} + \sum\limits_{(A, Y) \in E_P} x_{A, Y} + \sum\limits_{(X,A) \in E'_P} z_{X, A} \ge 1$ \hspace{10mm} $\forall A \in \V_P$ \\
        \>\>\>\> $x_{X, Y} \ge 0$  \hspace{10mm} $\forall (X, Y) \in E_P$ \\
        \>\>\>\> $z_{X, Y} \ge 0$  \hspace{10mm} $\forall (X, Y) \in E'_P$
    }
    LP${^{(+)}_2}$ differs from LP${^{(+)}_1}$ in that the former drops the variables $z_{X,Y}$ of those edges $(X, Y) \in E_P \setminus E'_P$. This happens both in the constraints and the objective function.
    \minipg{0.9\linewidth}{
        {\bf Claim 2:} LP$^{(+)}_1$ and LP$^{(+)}_2$ have the same optimal value.
    }
    To prove the claim, first observe that, given a feasible solution $\set{x_{X, Y} \mid (X, Y) \in E_P} \cup \set{z_{X, Y} \mid (X,$ $Y)$ $\in E'_P}$ to LP$^{(+)}_2$, we can extend it into a feasible solution to LP$^{(+)}_1$ by padding $Z_{X, Y}$ $=$ $0$ for each $(X, Y) \in E_P \setminus E'_P$. Hence, the optimal value of LP$^{(+)}_1$ cannot exceed that of LP$^{(+)}_2$. On the other hand, as mentioned before, $\set{x^*_{X, Y}, z^*_{X, Y} \mid (X, Y) \in E_P}$ is an optimal solution to LP$^{(+)}_1$. In this solution, $z^*_{X, Y} = 0$ for every $(X, Y) \in E_P \setminus E'_P$. Thus, $\set{x^*_{X, Y} \mid (X, Y) \in E_P} \cup \set{z^*_{X, Y} \mid (X, Y) \in E'_P}$ makes a feasible solution to LP$^{(+)}_2$, achieving the same objective function value as the optimal value of LP$^{(+)}_1$. Claim 2 now follows.

    \vgap

    Finally, notice that LP$^{(+)}_2$ is exactly the dual modular LP defined by $\dc'$. Hence, $\log (\polym(\dc'))$ is exactly the optimal value of LP$^{(+)}_2$. Thus, $\polym(\dc') = O(\polym(\dc))$ can now be derived from the above discussion and the fact that $\log(\polym(\dc))$ is asymptotically the optimal value of LP$^{(+)}$.
\end{proof}

\noindent {\bf Case $\bm{\lambda \le \sqrt{m}}$.} Let us first define several concepts. A {\em directed star} refers to a directed graph where there are $t \ge 2$ vertices, among which one vertex, designated the {\em center}, has $t-1$ edges (in-coming and out-going edges combined), and every other vertex, called a {\em petal}, has only one edge (which can be an in-coming or out-going edge). Now, consider a directed bipartite graph between $U_1$ and $U_2$, each being an independent sets of vertices (an edge may point from one vertex in $U_1$ to a vertex in $U_2$, or vice versa). A {\em directed star cover} of the bipartite graph is a set of directed stars such that
\myitems{
    \item each directed star is a subgraph of the bipartite graph,
    \item no two directed stars share a common edge, and
    \item every vertex in $U_1 \cup U_2$ appears in exactly one directed star.
}
A directed star cover is {\em minimum} if it has the least number of edges, counting all directed stars in the cover.

\vgap

Next, we review an expression about $\polym(\dc)$ derived in \cite{jrr21}. Find all the strongly connected components (SCCs) of $G_\dc = (V_P, E_P)$. Adopting terms from \cite{jrr21}, an SCC is classified as (i) a {\em source} if it has no in-coming edge from another SCC, or a {\em non-source} otherwise; (ii) {\em trivial} if it consists of a single vertex, or {\em non-trivial} otherwise. Define:
\myitems{
    \item $S$ = the set of vertices in $G_\dc$ each forming a trivial source SCC by itself.
    \item $T$ = the set of vertices in $G_\dc$ receiving an in-coming edge from at least one vertex in $S$.
}
Take a minimum directed star cover of the directed bipartite graph induced by $S$ and $T$. Define
\myitems{
    \item $S_1$ = the set of vertices in $S$ each serving as the center of some directed star in the cover.
    \item $S_2 = S \setminus S_1$.
    \item $T_2$ = the set of vertices in $T$ each serving as the center of some directed star in the cover.
    \item $T_1 = T \setminus T_2$.
}
Note that the meanings of the symbols $S_1, S_2, T_1$, and $T_2$ follow exactly those in \cite{jrr21} for the reader's convenience (in particular, note the semantics of $T_1$ and $T_2$).

\vgap

We now introduce three quantities:
\myitems{
    \item $c_1$: the number of non-trivial source SCCs;
    \item $n_1$: the total number of vertices in non-trivial source SCCs;
    \item $n_2 = |V_P| - n_1 - |S| - |T|$.
}
Jayaraman et al.\ \cite{jrr21} showed:
\myeqn{
    \polym_\dir(m, \lambda, P)
    &=&
    \Theta \Big(
        m^{c_1 + |S|} \cdot \lambda^{n_1 + n_2 + |T_1| - 2c_1 - |S_1|}
    \Big).
    \label{eqn:app:polym-dir-jrr21}
}

Let $G'_\dc = (V'_P, E_P')$ be an arbitrary weakly-connected acyclic subgraph of $G_\dc$ satisfying all the conditions below.
\myitems{
    \item $V_P = V_P'$.

    \item $E_P'$ contains all the edges in the minimum directed star cover identified earlier.

    \item In each non-trivial source SCC, every vertex, except for one, has one in-coming edge included in $E'_P$. We will refer to the vertex $X$ with no in-coming edges in $E'_P$ as the SCC's {\em root}. The fact that every other vertex $Y$ in the SCC has an in-coming edge in $E'_P$ implies $(X, Y) \in E'_P$ for at least one $Y$. We designate one such $(X, Y)$ as the SCC's {\em main edge}.

    \item In each non-trivial non-source SCC, every vertex has an in-coming edge included in $E'_P$.
}
It is rudimentary to verify that such a subgraph $G'_\dc$ must exist.

\vgap

From $G_\dc = (V_P, E_P)$ and $G'_\dc = (V_P', E_P')$, we create a set $\dc'$ of degree constraints as follows.
\myitems{
    \item For each edge $(X, Y) \in E_P$ (note: not $E_P'$), add a constraint $(\emptyset, \set{X, Y}, m)$ to $\dc'$.

    \item We inspect each directed star in the minimum directed star cover and distinguish two possibilities.
    \myitems{
        \item Scenario 1: The star's center $X$ comes from $S_1$. Let the star's petals be $Y_1, Y_2, ..., Y_t$ for some $t \ge 1$; the ordering of the petals does not matter. For each $i \in [t-1]$, we add a constraint $(\set{X}, \set{X, Y_i}, \lambda)$ to $\dc'$. We will refer to $(X, Y_t)$ as the star's {\em main edge}.

        \item Scenario 2: The star's center $X$ comes from $T_2$. Nothing needs to be done.
    }

    \item Consider now each non-trivial source SCC. Remember that every vertex $Y$, other than the SCC's root, has an in-coming edge $(X, Y) \in E_P'$. For every such $Y$, if $(X, Y)$ is not the SCC's main edge, add a constraint $(\set{X}, \set{X, Y}, \lambda)$ to $\dc'$.

    \item Finally, we examine each non-source SCC. As mentioned, every vertex $Y$ in such an SCC has an in-coming edge $(X, Y) \in E_P'$. For every $Y$, add a constraint $(\set{X}, \set{X, Y}, \lambda)$ to $\dc'$.
}

The rest of the proof will show $\polym(\dc') = \Theta(\polym(\dc))$. As $\dc' \subset \dc$, we must have $\polym(\dc') \ge \polym(\dc)$ following the same reasoning used in the $\lambda > \sqrt{m}$ case.

\vgap

We will now proceed to argue that $\polym(\dc') = O(\polym(\dc))$. Recall that $\log (\polym(\dc'))$ is the optimal value of the dual modular LP of $\dc'$ (see Section~\ref{sec:pre}). On the other hand, the value of $\polym(\dc)$ satisfies \eqref{eqn:app:polym-dir-jrr21}. In the following, we will construct a feasible solution to the dual modular LP of $\dc'$ under which the LP's objective function achieves the value of
\myeqn{
    \Big((c_1 + |S|) \cdot \log m \Big)
    +
    (n_1 + n_2 + |T_1| - 2c_1 - |S_1|) \cdot \log \lambda
    \label{eqn:app:throwaway-5}
}
which will be sufficient for proving Lemma~\ref{lmm:app:throwaway:opt-z-acyclic}.

\vgap

The dual modular LP associates every constraint $(\X, \Y, N_{\Y\mid\X}) \in \dc'$ with a variable $\delta_{\Y \mid \X}$. We determine these variables' values as follows.
\myitems{
    \item For every constraint $(\X, \Y, N_{\Y\mid\X}) \in \dc'$ where $N_{\Y\mid\X} = \lambda$, set $\delta_{\Y \mid \X} = 1$.
    \item Consider each directed star in the minimum directed star.
    \myitems{
        \item Scenario 1: The star's center $X$ comes from $S_1$. For the star's main edge $(X, Y)$, the constraint $(\emptyset, \set{X, Y}, m)$ exists in $\dc'$. Set $\delta_{\set{X, Y} \mid \emptyset} = 1$.

        \item Scenario 2: The star's center $X$ comes from $T_2$. For every petal $Y$ of the star, the constraint $(\emptyset, \set{X, Y}, m)$ exists in $\dc'$. Set $\delta_{\set{X, Y} \mid \emptyset} = 1$.
    }
    \item Consider each non-trivial source SCC. Let $(X, Y)$ be the main edge of the SCC. The constraint $(\emptyset, \set{X, Y}, m)$ exists in $\dc'$. Set $\delta_{\set{X, Y} \mid \emptyset} = 1$.
}
The other variables that have not yet been mentioned are all set to 0.

\vgap

It is tedious but straightforward to verify that all the constraints of the dual modular LP are fulfilled. To confirm that the objective function indeed evaluates to \eqref{eqn:app:throwaway-5}, observe:
\myitems{
    \item There are $c_1 + |S|$ constraints of the form $(\emptyset, \set{X, Y}, m)$ with $\delta_{\set{X,Y} \mid \emptyset} = 1$. Specifically, $c_1$ of them come from the roots of the non-trivial source SCCs, $|S_1|$ of them come from the star center vertices in $S_1$, and $|S_2|$ of them come from the petal vertices in $S_2$.

    \item There are $n_1 + n_2 + |T_1| - 2c_1 - |S_1|$ of the form $(\set{X}, \set{X, Y}, \lambda)$ with $\delta_{\set{X,Y} \mid \set{X}} = 1$. Specifically, $n_1 - 2c_1$ of them come from the non-main edges of the non-trivial source SCCs, $n_2$ of them come from the vertices that are not in any non-trivial source SCC and are not in $S \cup T$, and $|T_1|-|S_1|$ of them come from the petal vertices that are (i) in $T_1$ but (ii) not in the main edges of their respective stars.
}

We now conclude the whole proof of Lemma~\ref{lmm:directed-throwaway-constraints}.

\section{Proof of Lemma \ref{lmm:polym-undirected-closed} } \label{app:proof:polym-undirected-closed}

By definition, $\log (\polym_\undir(m,\lambda,P))$ equals the optimal value of the following LP:
\mytab{
    \> {\bf Polymatroid LP:} maximize $h(V_P)$, from all set functions $h(\cdot)$ over $V_P$, subject to\\
    \> (I) \> \hspace{2mm}$h(\emptyset) = 0$ \\
    \> (II) \> \hspace{2mm}$h(\set{X, Y}) \leq \log m$ $\quad$ $\forall \set{X, Y} \in E_P$ \\
    \> (III) \> \hspace{2mm}$h(\set{X, Y})-h(\set{X}) \leq \log \lambda$ $\quad$ $\forall \set{X, Y} \in E_P$ \\
    \> (IV) \> \hspace{2mm}$h(\X \cup \Y)+h(\X \cap \Y) \leq h(\X) + h(\Y)$ $\quad$ $\forall \X, \Y \subseteq V_P$ \\
    \> (V) \> \hspace{2mm}$h(\X) \leq h(\Y)$ $\quad$ $\forall \X \subseteq \Y \subseteq V_P$
}
\noindent To see that the above is an instance of linear programming, one can view a set function $h$ over $V_P$ as a point in a $2^{|V_P|}$-dimensional space, where each dimension is a different subset of $V_P$. Consequently, for any subset $\X \subseteq V_P$, $h(\X)$ can be regarded as the ``coordinate'' of this point on the dimension $\X$.

\subsection{Upper Bounding the Objective Function}

First, we will show that \eqref{eqn:polym-undirected-closed} is an upper bound of $\polym_\undir(m,\lambda,P)$. More precisely, for every feasible solution $h(.)$ to the polymatroid LP, we will prove:
\myeqn{
        h(V_P) &\le& \log m + (k-2) \log \lambda \label{eqn:polymat:upper1} \\
        h(V_P) &\le& (k_\mit{cycle}/2+\beta)\log m + (k_\mit{star} - 2\beta) \log \lambda. \label{eqn:polymat:upper2}
}
This implies that $\exp_2 (h(V_P))$ is always bounded by the right hand side of \eqref{eqn:polym-undirected-closed} and, hence, so is $\polym_\undir(m, \lambda, P)$.

\extraspacing{\bf Proof of \eqref{eqn:polymat:upper1}}.
Let us order the vertices in $V_P$ as $A_1, A_2, ..., A_k$ with the property that, for each $i \in [2, k]$, the vertex $A_i$ is
adjacent in $P$ to at least one vertex $A_j$ with $j < i$. We will denote this value of $j$ as $\back(i)$. Such an ordering
definitely exists because $P$ is connected. For $i \ge 3$, define
$A_{[i]} = \set{A_1, A_2, ...,  A_i}$.

To start with, let us observe that, by the constraint (IV) of the polymatroid LP, the inequality
\myeqn{
    h(A_{[i]}) + h(A_{\back(i)}) \le h(A_{[i-1]}) + h(A_{\back(i)}, A_i)
    \label{eqn:polymat:upper1-help}
}
holds for all $i \in [2, k]$. Thus:
\myeqn{
    h(V_P) &=& h(A_1, A_2) + \sum_{i=3}^k h(A_{[i]}) - h(A_{[i-1]}) \nn \\
    \explain{by \eqref{eqn:polymat:upper1-help}}
    &\le& h(A_1, A_2) + \sum_{i=3}^k h(A_{\back(i)}, A_i) - h(A_{\back(i)}) \nn \\
    &\le& \log m + (k - 2) \log \lambda. \nn
}
The last step used constraints (II) and (III) of the polymatroid LP.

\extraspacing {\bf Proof of \eqref{eqn:polymat:upper2}.} For each cycle $\mycycle_i$ ($i \in [\alpha]$) in the fractional edge-cover decomposition of $P$, define $V(\mycycle_i)$ as the set of vertices in $\mycycle_i$ and set $k(\mycycle_i) = |V(\mycycle_i)|$. Likewise, for each star $\mystar_j$ ($j \in [\beta]$), define $V(\mystar_j)$ as the set of vertices in $\mystar_j$ and set $k(\mystar_j) = |V(\mystar_j)|$. We aim to prove

\myeqn{
    && \text{for each $i \in [\alpha]$: }
    h(V(\mycycle_i)) \le (k(\mycycle_i) / 2) \log m; \label{eqn:polymat:upper2-help1} \\
    && \text{for each $j \in [\beta]$: }
    h(V(\mystar_j)) \le \log m + (k(\mystar_j) - 2) \log \lambda.  \label{eqn:polymat:upper2-help2}
}

Once this is done, we can establish \eqref{eqn:polymat:upper2} as follows. First, from constraint (IV), we know $h(\X \cup \Y) \le h(\X) + h(\Y)$ for any disjoint $\X, \Y \subseteq V_P$; call this {\em the disjointness rule}. Then, we can derive
\myeqn{
    h(V_P) &\le& \sum_{i=1}^\alpha h(V(\mycycle_i)) + \sum_{j=1}^\beta h(V(\mystar_j))
     \hspace{2mm} \explain{disjointness rule} \nn \\
    \explain{by \eqref{eqn:polymat:upper2-help1}, \eqref{eqn:polymat:upper2-help2}}
    &\le&
    \sum_{i=1}^\alpha \fr{k(\mycycle_i)}{2} \log m + \sum_{j=1}^\beta (\log m + (k(\mystar_j) - 2) \log \lambda) \nn \\
    &=&
    (k_\mit{cycle}/2+\beta)\log m + (k_\mit{star} - 2\beta) \log \lambda \nn
}
as desired.

We proceed to prove \eqref{eqn:polymat:upper2-help1}. Let us arrange the vertices of $\mycycle_i$ in clockwise order as $X_1, X_2, ..., X_{k(\mycycle_i)}$. For $t \ge 3$, define $X_{[t]} = \set{X_1, X_2, ...,  X_t}$. Applying the disjointness rule, we get
\myeqn{
    h(V(\mycycle_i)) &\le&
    h(X_1, X_{k(\mycycle_i)}) + \sum_{t=2}^{k(\mycycle_i) - 1} h(X_t).
    \label{eqn:polymat:upper2-help3}
}
Equipped with the above, we can derive
\myeqn{
    \hspace{-6mm} &&
    k(\mycycle_i) \cdot \log m \hspace{2mm} \explain{by constraint (II)} \nn \\
    \hspace{-6mm} &&
    \ge
    \Big(\sum_{t=1}^{k(\mycycle_i) - 1} h(X_t, X_{t+1})\Big)
    +
    h(\set{X_{k(\mycycle_i)}, X_1}) \nn \\
    \hspace{-6mm} &&
    \explain{the next few steps will apply constraint (IV)} \nn \\
    \hspace{-6mm} &&
    \ge h(X_{[3]}) + h(X_2) +
    \Big(\sum_{t=3}^{k(\mycycle_i) - 1} h(X_t, X_{t+1})\Big)
    + h(X_{k(\mycycle_i)}, X_1) \nn
}
\myeqn{
    \hspace{-6mm} &&
    \ge
    h(X_{[4]}) + h(X_2) + h(X_3) +
    \Big(\sum_{t=4}^{k(\mycycle_i) - 1} h(X_t, X_{t+1})\Big)
    +
    h(X_{k(\mycycle_i)}, X_1) \nn \\
    \hspace{-6mm} && ... \nn \\
    \hspace{-6mm} && \ge
    h(X_{[k(\mycycle_i)]}) + \Big(\sum_{t=2}^{k(\mycycle_i) - 1} h(X_t) \Big) +
    h(X_{k(\mycycle_i)}, X_1)
    \nn \\
    \hspace{-6mm} &&
    \explain{the next step applies \eqref{eqn:polymat:upper2-help3} and $V(\mycycle_i) = X_{[k(\mycycle_i)]}$} \nn \\
    \hspace{-6mm} &&
    \ge 2 \cdot h(V(\mycycle_i))
    \hspace{5mm} \nn
}
as claimed in \eqref{eqn:polymat:upper2-help1}.

It remains to prove \eqref{eqn:polymat:upper2-help2}. Let us label the vertices in $\mystar_j$ as $Y_1, Y_2, ...,$ $Y_{k(\mystar_j)}$ with $Y_1$ being the center vertex. For $t \ge 3$, define $Y_{[t]} = \set{Y_1,$ $Y_2, ...,  Y_t}$. Then:
\myeqn{
    h(V(\mystar_j)) &=&
    h(Y_1, Y_2) + \sum_{t=3}^{k(\mystar_j)} h(Y_{[t]}) - h(Y_{[t-1]}) \nn \\
    \explain{by constraint (IV)}
    &\le&
    h(Y_1, Y_2) + \sum_{t=3}^{k(\mystar_j)} (h(Y_1, Y_t) - h(Y_1)) \nn \\
    \explain{by constraints (II), (III)}
    &\le&
    \log m + (k(\mystar_j) - 2) \log \lambda \nn
}
as claimed in \eqref{eqn:polymat:upper2-help2}.

\subsection{Constructing an Optimal Set Function}

To prove Lemma~\ref{lmm:polym-undirected-closed}, we still need to show that $\polym_\undir(m,\lambda,P)$ is at least the right hand side of \eqref{eqn:polym-undirected-closed}. For this purpose, it suffices to prove (i) when $\lambda \le \sqrt{m}$, there is a feasible set function $h^*$ whose $h^*(V_P)$ equals the right hand side of \eqref{eqn:polymat:upper1}, and (ii) when $\lambda > \sqrt{m}$, there is a feasible set function $h^*$ whose $h^*(V_P)$ equals the right hand side of \eqref{eqn:polymat:upper2}. This subsection will construct these set functions explicitly.

\extraspacing {\bf Case $\lambda \le \sqrt{m}$.} In this scenario, the function $h^*$ is easy to design. For each $\X \subseteq V_P$, set
\myeqn{
    h^*(\X) =
    \begin{cases}
        0 & \text{if $\X = \emptyset$} \\
        \log m + (|\X| - 2) \log \lambda & \text{otherwise}
    \end{cases}
    \label{eqn:polymat:h*-small-lambda}
}
Obviously, $h^*(V_P) = \log m + (|V_P| - 2) \log \lambda$, as needed. It remains to explain why this $h^*$ is a feasible solution to the polymatroid LP. We prove as follows.

\vgap

Constraints (I), (II), and (III) are trivial to verify and omitted. Regarding (IV), first note that the constraint is obviously true if $\X$ or $\Y$ is empty. If neither of them is empty, we can derive:
\myeqn{
    h^*(\X\cup \Y) + h^*(\X \cap \Y)
    \hspace{-2mm} &=&
    \hspace{-2mm} 2\log m + (|\X \cup \Y| + |\X \cap \Y|-4) \log \lambda \nn \\
        &=&
        \hspace{-2mm} 2\log m + (|\X| + |\Y| - 4) \log \lambda \nn \\
        &=&
        \hspace{-2mm} h^*(\X) + h^*(\Y) \label{eqn:app:h*-verification-help1}
}
as needed. Now, consider constraint (V). If $\X = \emptyset$, the constraint holds because $h^*(\Y) = \log m + (|\Y|-2) \log \lambda \ge \log m - \log \lambda \ge 0$. If $X \ne \emptyset$, then
\myeqn{
    h^*(\X) &=& \log m + (|\X| - 2) \log \lambda \nn \\
    &\leq& \log m + (|\Y| - 2) \log \lambda \nn \\
    &=& h^*(\Y). \nn
}

\extraspacing {\bf Case $\lambda > \sqrt{m}$.} Let us look at the fractional edge-cover decomposition of $P$:  $(\mycycle_1, ..., \mycycle_\alpha$, $\mystar_1, ...,$  $\mystar_\beta)$. As before, for each $i \in [\alpha]$, define $V(\mycycle_i)$ as the set of vertices in the cycle $\mycycle_i$; for each $j \in [\beta]$, define $V(\mystar_j)$ as the set of vertices in the star $\mystar_j$.

\vgap

To design the function $h^*$, for each $\X \subseteq V_P$, we choose the value $h^*(\X)$ by the following rules.
\myitems{
    \item (Rule 1): If $\X = \emptyset$, then $h^*(\X) = 0$.
    \item (Rule 2): if $\X \subseteq V(\mycycle_i)$ for some $i \in [\alpha]$ but $\X \ne \emptyset$, then $h^*(\X) = \fr{|\X|}{2} \log m$.

    \item (Rule 3): if $\X \subseteq V(\mystar_j)$ for some $j \in [\beta]$ but $\X \ne \emptyset$, then $h^*(\X) = \log m + (|\X| - 2) \log \lambda$.

    \item (Rule 4): Suppose that none of the above rules applies. For each $i \in [\alpha]$, define $\Y_i = \X \cap V(\mycycle_i)$; similarly, for each $j \in [\beta]$, define $\Z_j = \X \cap V(\mystar_j)$. Then:
    \myeqn{
        h^*(\X) &=& \sum_{i=1}^\alpha h^*(\Y_i) + \sum_{j=1}^\beta h^*(\Z_j). \label{eqn:polymat:h*-large-lambda-rule3}
    }
    The above equation is well-defined because each $h^*(\Y_i)$ can be computed using Rules 1 and 2, and each $h^*(\Z_j)$ can be computed using Rules 1 and 3.
}
As a remark, our construction ensures that \eqref{eqn:polymat:h*-large-lambda-rule3} holds for all $\X \subseteq V_P$.

\vgap

It is easy to check that $h^*(V_P) = \sum_{i=1}^\alpha h^*(V(\mycycle_i)) + \sum_{j=1}^\beta h^*(V(\mystar_j))$, which is $(k_\mit{cycle}/2+\beta)\log m + (k_\mit{star} - 2\beta) \log \lambda$, as needed. It suffices to verify the feasibility of $h^*$.

\vgap

Constraint (I) is guaranteed by Rule 1. Next, we will discuss constraints (II) and (III) together. The verification is easy (and hence omitted) if $\set{X, Y}$ is a cycle edge or a star edge. Now, consider that $\set{X, Y}$ is neither a cycle edge nor a star edge. By the properties of fractional edge-cover decomposition, one of the following must occur:
\myitems{
    \item (C-1) $X$ and $Y$ are in two different cycles;
    \item (C-2) $X$ and $Y$ are in two different stars;
    \item (C-3) one of $X$ and $Y$ is in a cycle and the other is in a star.
}
In all the above scenarios, it holds that $h^*(\set{X,Y}) = h^*(\set{X}) + h^*(\set{Y})$. Thus, to confirm (II), it suffices to show $h^*(\set{X}) + h^*(\set{Y}) \le \log m$, and to confirm (III), it suffices to show $h(\set{Y}) \le \log \lambda$. It is rudimentary to verify both of these inequalities using Rules 2 and 3 and the fact $\log m < 2 \log \lambda$.

\vgap

Constraint (IV) is trivially true if either $\X$ or $\Y$ is empty. Now, assume that neither is empty. If $\X, \Y \subseteq \V(\mycycle_i)$ for some $i \in [\alpha]$, we have:
\begin{equation*}
    \begin{aligned}
        h^*(\X \cup \Y) + h^*(\X \cap \Y) &= \frac{|\X \cap \Y| + |\X \cup \Y|}{2}\log m \\
        &= \frac{|\X| + |\Y|}{2}\log m\\
        &= h^*(\X) + h^*(\Y).
    \end{aligned}
\end{equation*}
If $\X, \Y \subseteq \V(\mystar_j)$ for some $j \in [\beta]$, the reader can verify (IV) with the same derivation in \eqref{eqn:app:h*-verification-help1}.

\vgap

Next, we consider arbitrary $\X, \Y \subseteq \V$. Define for each $i \in [\alpha]$ and $j \in [\beta]$:
\myeqn{
    \X^{\mit{cycle}}_i &=& \X \cap \V(\mycycle_i) \label{eqn:app:h*-verification-help4} \\
    \X^{\mit{star}}_j &=& \X \cap \V(\mystar_j) \label{eqn:app:h*-verification-help5} \\
    \Y^{\mit{cycle}}_i &=& \Y \cap \V(\mycycle_i) \label{eqn:app:h*-verification-help6} \\
    \Y^{\mit{star}}_j &=& \Y \cap \V(\mystar_j). \label{eqn:app:h*-verification-help7}
}
We can derive:
\myeqn{
    &&h^*(\X \cup \Y) \nn \\
    \explain{by Rule 4} &=& \sum_{i=1}^\alpha h^*((\X \cup \Y) \cap \V(\mycycle_i)) + \sum_{j=1}^\beta h^*((\X \cup \Y) \cap \V(\mystar_j))\nn \\
    &=& \sum_{i=1}^\alpha h^*(\X^\mit{cycle}_i \cup \Y^\mit{cycle}_i) + \sum_{j=1}^\beta h^*(\X^\mit{star}_j \cup \Y^\mit{star}_j) \nn
}
and similarly:
\myeqn{
    &&h^*(\X \cap \Y) \nn \\
    &=& \sum_{i=1}^\alpha h^*((\X \cap \Y) \cap \V(\mycycle_i)) + \sum_{j=1}^\beta h^*((\X \cap \Y) \cap \V(\mystar_j))\nn \\
    &=& \sum_{i=1}^\alpha h^*(\X^\mit{cycle}_i \cap \Y^\mit{cycle}_i) + \sum_{j=1}^\beta h^*(\X^\mit{star}_j \cap \Y^\mit{star}_j). \nn
}
Recall that (IV) has been validated in the scenario where (i) $\X$ and $\Y$ are contained in the same cycle, or (ii) they are contained in the same star. Thus, it holds for each $i \in [\alpha]$ and $j \in [\beta]$ that
\myeqn{
    \hspace{-4mm}
    h^*(\X^{\mit{cycle}}_i \cup \Y^{\mit{cycle}}_i) + h^*(\X^{\mit{cycle}}_i \cap \Y^{\mit{cycle}}_i) \leq h^*(\X^{\mit{cycle}}_i) + h^*(\Y^{\mit{cycle}}_i) \label{eqn:app:h*-verification-help2}\\
    \hspace{-4mm}
    h^*(\X^{\mit{star}}_j \cup \Y^{\mit{star}}_j) + h^*(\X^{\mit{star}}_j \cap \Y^{\mit{star}}_j) \leq h^*(\X^{\mit{star}}_j) + h^*(\Y^{\mit{star}}_j). \label{eqn:app:h*-verification-help3}
}
Equipped with these facts, we get:
\myeqn{
    &&h^*(\X \cup \Y) + h^*(\X \cap \Y) \nn \\
    &=& \sum_{i=1}^{\alpha} h^*(\X^{\mit{cycle}}_i \cup \Y^{\mit{cycle}}_i) + h^*(\X^{\mit{cycle}}_i \cap \Y^{\mit{cycle}}_i) + \nn \\
    &&\sum_{j=1}^{\beta} h^*(\X^{\mit{star}}_j \cup \Y^{\mit{star}}_j) + h^*(\X^{\mit{star}}_j \cap \Y^{\mit{star}}_j) \nn \\
    \explain{by \eqref{eqn:app:h*-verification-help2} and \eqref{eqn:app:h*-verification-help3}} &\leq& \sum_{i=1}^{\alpha} h^*(\X^{\mit{cycle}}_i) + h^*(\Y^{\mit{cycle}}_i) + \sum_{j=1}^{\beta} h^*(\X^{\mit{star}}_j) + h^*(\Y^{\mit{star}}_j) \nn \\
    \explain{by Rule 4} &=& h^*(\X) + h^*(\Y).\nn
}
This verifies the correctness of constraint (IV).

\vgap

Finally, let us look at (V). This constraint is trivially met if $\X = \emptyset$. Next, we assume that $\X$ is not empty. If $\Y$ is a subset of $\V(\mycycle_i)$ for some $i \in [\alpha]$, then, by Rule 2, $h^*(\X) = \frac{|\X|}{2} \log m \leq \frac{|\Y|}{2} \log m = h^*(\Y)$. If $\Y$ is a subset of $\V(\mystar_{j})$ for some $j \in [\beta]$, then, by Rule 3, $h^*(\X) = \log m + (|\X|-2) \log \lambda \leq \log m + (|\Y|-2) \log \lambda = h^*(\Y)$.

\vgap

It remains to consider the situation where $\Y$ can be any subset of $\V$. Using the definitions in \eqref{eqn:app:h*-verification-help4}-\eqref{eqn:app:h*-verification-help7}, we have:
\myeqn{
    h^*(\X) &=& \sum_{i=1}^{\alpha} h^*(\X^{\mit{cycle}}_i) + \sum_{j=1}^{\beta} h^*(\X^{\mit{star}}_j) \hspace{3mm} \explain{by Rule 4} \nn \\
    &\leq& \sum_{i=1}^{\alpha} h^*(\Y^{\mit{cycle}}_i) + \sum_{j=1}^{\beta} h^*(\Y^{\mit{star}}_j) \nn \\
    && \textrm{(we have verified (V) in the scenario where $\Y$ is contained in a cycle or a star)} \nn \\
    \explain{by Rule 4} &=& h^*(\Y)  \nn
}
as needed to verify constraint (V).

\section{Equation \eqref{eqn:polym-undirected-closed} as an Output Size Bound and Its Tightness} \label{app:tightness-undirected}

We will start by proving that \eqref{eqn:polym-undirected-closed} is an asymptotic upper bound on $|\occ(G,P)|$ for any graph $G$ that has $m$ edges and a maximum vertex-degree at most $\lambda$. First, we will prove that $|\occ(G,P)| = O(m \cdot \lambda^{k-2})$ (recall that $k$ is the number of vertices in $P$). Identify an arbitrary spanning tree $\T$ of the pattern $P$. It is rudimentary to verify that $\T$ has $O(m \cdot \lambda^{k-2})$ occurrences in $G$,\footnote{There are $O(m)$ choices to map an arbitrary edge $\T$ to an edge in $G$, and then $\lambda$ choices to map each of the $k-2$ remaining vertices of $\T$ to a vertex in $G$.} implying that the number of occurrences of $P$ is also $O(m \cdot \lambda^{k-2})$. Next, we will demonstrate that $|\occ(G,P)| = O(m^{\fr{k_\mit{cycle}}{2} + \beta} \cdot \lambda^{k_\mit{star}-2\beta})$. For each $i \in [\alpha]$, let $k(\mycycle_i)$ be the number of vertices in $\mycycle_i$; for each $j \in [\beta]$, let $k(\mystar_j)$ be the number of vertices in $\mystar_j$. The fractional edge cover number of $\mycycle_i$ ($i \in [\alpha]$) is $\rho^*(\mycycle_i) = k(\mycycle_i)/2$; this means $\sum_{i=1}^k \rho^*(\mycycle_i) = k(\mycycle_i)/2 = k_\mit{cycle}/2$. For each $i \in [\alpha]$, the pattern $\mycycle_i$ can have $O(m^{\rho^*(\mycycle_i)})$ occurrences in $G$. For each $j \in [\beta]$, by our earlier analysis, the pattern $\mystar_j$ can have $O(m^{k(\mystar_j) - 2})$ occurrences in $G$. Thus, the number of occurrences of $P$ must be asymptotically bounded by
\myeqn{
    \prod_{i=1}^\alpha m^{\rho^*(\mycycle_i)} \cdot \prod_{j=1}^\beta m \cdot \lambda ^{k(\mystar_j) - 2}
    =
    m^{\fr{k_\mit{cycle}}{2} + \beta} \cdot \lambda^{k_\mit{star}-2\beta}. \nn
}

The rest of this section will concentrate on the tightness of \eqref{eqn:polym-undirected-closed} as an upper bound of $|\occ(G,P)|$. Our objective is to prove:

\begin{theorem} \label{thm:tightness-undirected}
    Fix an arbitrary integer $k \ge 2$. For any values of $m$ and $\lambda$ satisfying $m \geq \max\{16k^2,64\}$ and $k \leq \lambda \leq m/(4k)$, there is an undirected simple graph $G^*$ satisfying all the conditions below:
    \myitems{
        \item $G^*$ has at most $m$ edges and and a maximum vertex degree at most $\lambda$;
        \item For any undirected simple pattern graph $P = (V_P, E_P)$ that has $k$ vertices, the number of occurrences of $P$ in $G^*$ is at least $\fr{1}{(4k)^k} \cdot \polym_\undir(m,\lambda,P)$, where $\polym_\undir(m,\lambda,P)$ is given in \eqref{eqn:polym-undirected-closed}.
    }
\end{theorem}

It is worth noting that we aim to construct a {\em single} $G^*$ that is a ``bad'' input for {\em all} possible patterns (with $k$ vertices) simultaneously. Our proof will deal with the case $k \leq \lambda < \sqrt{m}$ and the case $\sqrt{m}\leq \lambda \leq m/(4k)$ separately.

\vgap

\extraspacing {\bf Remark.} The main challenge in the proof is to establish the factor $\fr{1}{(4k)^k}$. There exist lower bound arguments \cite{jrr21,s23} that can be used in our context to build a hard instance for each {\em individual} pattern $P$. A naive way to form a {\em single} hard input $G^*$ is to combine the hard instances of all possible patterns having $k$ vertices. But this would introduce a gigantic factor roughly $1/2^{\Omega(k^2)}$.

\subsection{Case $k \leq \lambda < \sqrt{m}$}

In this scenario, $G^*$ only needs to be the graph that consists of $\lfloor m / \binom{\lambda}{2}\rfloor$ independent $\lambda$-cliques.\footnote{When $\lambda < \sqrt{m}$, $\binom{\lambda}{2} = \frac{\lambda^2-\lambda}{2} < m$. Hence, $G^*$ contains at least one clique.} An occurrence of $P$ can be formed by mapping $V_P$ to $k$ arbitrary distinct vertices in one $\lambda$-clique. The number of occurrences of $P$ in $G^*$ is at least
\myeqn{
    && \big \lfloor \frac{m}{{\lambda \choose 2}} \big \rfloor \cdot {\lambda \choose k} \nn \\
    \explain{applying ${\lambda \choose k} \ge (\lambda/k)^k$}
    &\geq&
    \Big(\frac{2m}{\lambda^2-\lambda}-1 \Big) \cdot (\lambda / k)^k \nn\\
    \explain{as $m > \lambda^2-\lambda$} \nn
    &>& \frac{m}{\lambda^2-\lambda} \cdot (\lambda /k)^k \nn\\
    \explain{as $k \ge 2$}
    &> &\frac{1}{k^k}\cdot m\cdot \lambda^{k-2}
    = \frac{1}{k^k}\cdot \polym_\undir(m,\lambda,P). \nn
}

\subsection{Case $\sqrt{m}\leq \lambda \leq m/(4k)$}

We construct $G^*$ as follows. First, create three disjoint sets of vertices: $V_A, V_B$, and $V_C$, whose sizes are $\lc \lambda/4 \rc$, $\lceil m/(4\lambda) \rceil$, and $\lceil \sqrt{m}/4 \rceil$. Because $\lambda \leq \frac{m}{4k}$ and $m \geq 16k^2$, each of $V_A$, $V_B$, and $V_C$ has a size at least $k$. Then, decide the edges of $G^*$ in the way  below:
\myitems{
    \item Add an edge between each pair of vertices in $V_B \cup V_C$ (thereby producing a clique of $\lceil m/(4\lambda) \rceil + \lceil \sqrt{m}/4 \rceil$ vertices).
    \item Add an edge $\set{u, v}$ for each vertex pair $(u, v) \in V_A \times V_B$.

}

Next, we show that $G^*$ has at most $m$ edges through a careful calculation. First, the number of edges between $V_A$ and $V_B$ is
\myeqn{
    &&\lceil \lambda/4 \rceil \cdot \lceil m/ (4\lambda)\rceil \nn\\
    &\leq& (\lambda/4+1)\cdot(m/ (4\lambda)+1)\nn\\
    &=& m/16 + \lambda/4 + m/(4\lambda) + 1\nn\\
    \explain{as $\sqrt{m} \leq \lambda \leq m$}&\leq& 5m/16 + \sqrt{m}/4 + 1 \nn\\
    \explain{as $16 \leq m$}&\leq& 7m/16. \nn
}
The number of edges between $V_B$ and $V_C$ is:
\myeqn{
    &&\lceil m/(4\lambda)\rceil \cdot \lceil  \sqrt{m}/4 \rceil \nn \\
    &\leq& (m/(4\lambda)+1) \cdot (\sqrt{m}/4+1)
    \nn\\
    &=& m^{1.5}/(16\lambda) + m/(4\lambda) + \sqrt{m}/4 + 1 \nn\\
    \explain{as $\sqrt{m} \leq \lambda$} & \leq& m/16 + \sqrt{m}/2 + 1\nn\\
    \explain{as $16 \leq m$}&\leq& m/4.\nn
}
The number of edges between vertices in $V_B$ is:
\myeqn{
    && \lceil m/(4\lambda)\rceil \cdot (\lceil m/(4\lambda)\rceil -1) / 2\nn\\
    &\leq& (m/(4\lambda) + 1)\cdot(m/(4\lambda)) / 2\nn\\
    &=& m^2/(32\lambda^2) + m/(8\lambda)\nn \\
    \explain{as $\sqrt{m} \leq \lambda$}&\leq& m/32 + \sqrt{m}/8 \leq  5m/32 \nn
}
The number of edges between vertices in $V_C$ is:
\myeqn{
    && \lceil \sqrt{m}/4\rceil \cdot (\lceil \sqrt{m}/4 -1) / 2\nn\\
    &\leq& (\sqrt{m}/4 + 1)\cdot(\sqrt{m}/4) / 2\nn\\
    &=& m/32 + \sqrt{m}/8 \leq 5m/32. \nn
}
Thus, the total number of edges in $G^*$ is at most $7m/16+m/4+5m/32+5m/32 = m$.

\vgap

The maximum vertex degree in $G^*$ is decided by the vertices in $V_B$ and equals
\myeqn{
    &&\lceil \lambda/4 \rceil + \lceil m/(4\lambda) \rceil -1 + \lceil \sqrt{m}/4 \rceil \nn \\
    &\leq& \lambda/4  +  m/(4\lambda)  +  \sqrt{m}/4 + 2 \nn \\
    \explain{as $\lambda\geq \sqrt{m}$} &\leq& 3\lambda/4 + 2 \leq \lambda. \nn
}
Therefore, $G^*$ satisfies the requirement in the first bullet of Theorem~\ref{thm:tightness-undirected}.

\vgap

The rest of the proof will focus on the theorem's second bullet. Let $P$ be an arbitrary pattern with $k$ vertices, and take a fractional edge-cover decomposition of $P$: $(\mycycle_1,$ $...,$ $\mycycle_\alpha,$ $\mystar_1,$ $..., \mystar_\beta)$. Define $k_\mit{star}$ and $k_\mit{cycle}$ in the same way as in \eqref{eqn:k_star} and \eqref{eqn:k_cycle}. We claim:

\begin{lemma} \label{lmm:app:divide_vertices}
    There is an integer $s \in [0, \beta]$ under which we can divide the vertices of $P$ into disjoint sets $U_A$, $U_B$ and $U_C$ such that
    \myitems{
        \item $|U_A| = k_\mit{star} - \beta - s$, $|U_B| = \beta - s$, $|U_C| = k_\mit{cycle} + 2s$;
        \item $P$ has no edges between $U_A$ and $U_C$;
        \item $P$ has no edges between any two vertices in $U_A$.
    }
\end{lemma}

The proof of the lemma is non-trivial and deferred to the end of this section. An occurrence of $P$ in $G^*$ can be formed through the three steps below:
\myenums{
    \item Map $U_A$ to $|U_A|$ distinct vertices in $V_A$.
    \item Map $U_B$ to $|U_B|$ distinct vertices in $V_B$.
    \item Map $U_C$ to $|U_C|$ distinct vertices in $V_C$.
}
Each step is carried out independently from the other steps. Hence, the number of occurrences of $P$ in $G^*$ is at least
\myeqn{
    &&\binom{\lambda/4}{k_{\mit{star}}-\beta-s} \cdot \binom{\fr{m}{4\lambda}}{\beta-s}\cdot \binom{\sqrt{m}/4}{k_{\mit{cycle}}+2s}\nn\\
    &\geq&
    \big(
        \fr{
            \lambda/4
        }{
            k_{\mit{star}}-\beta-s
        }
    \big)^{k_{\mit{star}}-\beta-s}
    \cdot
    \big(
        \fr{
            \fr{m}{4\lambda}
        }{
            \beta-s
        }
    \big)^{\beta-s}
    \cdot
    \big(
        \fr{
            \sqrt{m}/4
        }{
            k_\mit{cycle}+2s
        }
    \big)^{k_\mit{cycle}+2s}
    \nn\\
    &\geq& \frac{\lambda^{k_{\mit{star}}-\beta-s}}{4^{k_{\mit{star}}-\beta-s} \cdot k^{k_{\mit{star}}-\beta-s}} \cdot \frac{m^{\beta-s}}{4^{\beta-s} \cdot \lambda^{\beta-s} \cdot k^{\beta-s}} \cdot \frac{m^{k_{\mit{cycle}}/2+s}}{4^{k_{\mit{cycle}}+2s} \cdot k^{k_{\mit{cycle}}+2s}} \nn\\
    &=& \frac{1}{(4k)^k}\cdot m^{\frac{k_\mit{cycle}}{2}+\beta} \cdot \lambda^{k_\mit{star} -2 \beta}  =  \frac{1}{(4k)^k}\cdot \polym_\undir(m, \lambda, P).\nn
}

\extraspacing {\bf Proof of Lemma~\ref{lmm:app:divide_vertices}.} With each vertex $X$ in the pattern graph $P = (V_P, E_P)$, we associate a variable $\nu(X) \ge 0$ (equivalently, $\nu$ is a function from $V_P$ to $\real_{\ge 0}$). Consider the following LP defined on these variables.

\mytab{
    \> {\bf vertex-pack LP} \hspace{10mm} max $\sum_{X \in V_P} \nu(X)$  subject to \\[2mm]
    \>\>\>\> $\sum\limits_{X: X \in F} \nu(X) \leq 1$ \hspace{15mm} $\forall F \in E_P$ \\
    \>\>\>\> $ \nu(X) \geq 0$ \hspace{25.5mm} $\forall X \in V_P$
}

\noindent A feasible solution $\nu$ is said to be {\em half-integral} if $\nu(X) \in \set{0, \fr{1}{2}, 1}$ for each $X \in V_P$.

\vgap

Consider any fractional edge-cover decomposition of $P$: $(\mycycle_1,..., \mycycle_\alpha, \mystar_1, ..., \mystar_\beta)$. Recall that each star $\mystar_j$ ($j \in [\beta]$) contains a vertex designated as the center. We refer to every other vertex in the star as a {\em petal}.

\begin{lemma} \label{lmm:vertex-packing}
    For the vertex-pack LP, there exists a half-integral optimal solution $\set{\nu^*(X)\mid X\in V_P}$ satisfying all the conditions below.
    \myitems{
        \item $\nu^*(X) = 0.5$ for every vertex $X$ in $\mycycle_1, ... \mycycle_\alpha$.
        \item For every $j \in [\beta]$, the function $\nu^*$ has the properties below.
        \myitems{
            \item If $\mystar_j$ has at least two edges, then $\nu^*(X) = 0$ holds for the star's center $X$ and $\nu^*(Y) = 1$ holds for every petal $Y$.

            \item If $\mystar_j$ has only one edge $\set{X,Y}$ --- in which case we call $\mystar_j$ a ``one-edge star'' --- then $\nu^*(X) + \nu^*(Y) = 1$.
        }
    }
\end{lemma}
\begin{proof}
    Let us first define the dual of vertex-pack LP. Associate each edge $F \in E_P$ with a variable $w(F) \ge 0$ (equivalently, $w$ is a function from $E_P$ to $\real_{\ge 0}$). Then, the dual is:

    \mytab{
        \> {\bf edge-cover LP} \hspace{10mm} min $\sum_{F \in E_P} w(F)$  subject to \\[2mm]
        \>\>\>\> $\sum\limits_{F \in E_P: X \in F} w(F) \geq 1$ \hspace{10mm} $\forall X \in V_P$ \\
        \>\>\>\> $ w(F) \geq 0$ \hspace{26mm} $\forall F \in E_P$
    }

    The given fractional edge-cover decomposition implies an optimal solution $\set{w^*(F)\mid F\in E_p}$ to the edge-cover LP satisfying:
    \myitems{
        \item $w^*(F) = 0.5$ for every edge $F$ in $\mycycle_1, ... \mycycle_\alpha$;
        \item $w^*(F) = 1$ for every edge $F$ in $\mystar_1, ..., \mystar_\beta$;
        \item $w^*(F) = 0$ for every other edge $F \in E_p$.
    }
    By the complementary slackness theorem, the function $w^*$ implies an optimal solution $\set{\nu'(X)\mid X\in V_P}$ to the vertex-pack LP with the properties below.
    \myitems{
        \item {{\bf P1}:} For every edge $\set{X,Y}\in E_p$, if $w^*(\set{X,Y}) > 0$, then $\nu'(X) + \nu'(Y)= 1$.
        \item {{\bf P2}:} For every vertex $X$ satisfying $\sum_{X \in F} w^*(F)> 1$, $\nu'(X) = 0$.
    }

    From the above, we can derive additional properties of the solution $\set{\nu'(X)\mid X\in V_P}$.
    \myitems{
        \item {{\bf P3}:} $\nu'(X) = 0.5$ for every vertex $X$ in $\mycycle_1, ..., \mycycle_\alpha$. To see why, consider any $\mycycle_i$ ($i \in [\alpha]$). Let $X_1, X_2, ..., X_\ell$ (for some $\ell \ge 3$) be the vertices of $\mycycle_i$ in clockwise order. Property {\bf P1} yields $\ell$ equations: $\nu'(X_j) + \nu'(X_{j+1})= 1$ for $j \in [\ell - 1]$ and $\nu'(X_\ell) + \nu'(X_1)= 1$. Solving the system of these equations gives $\nu'(X_j) = 0.5$ for each $j \in [\ell]$.

        \item For every $j \in [\beta]$, we have:
        \myitems{
            \item{{\bf P4-1}:} If $\mystar_j$ has at least two edges, then $\nu'(X) = 0$ holds for the star's center $X$ and $\nu'(Y) = 1$ for every petal $Y$. To see why, notice that $\sum_{F\in E_P:X\in F} w^*(F)$ is precisely the number of edges in $\mystar_j$ and, hence, $\sum_{F\in E_P:X\in F} w^*(F) > 1$. Thus, {\bf P2} asserts $\nu'(X) = 0$. It then follows from {\bf P1} that $\nu'(Y) = 1-\nu'(X) = 1$ for every petal $Y$.

            \item{{\bf P4-2}:} If $\mystar_j$ has only one edge $\set{X,Y}$, then $\nu'(X) + \nu'(Y) = 1$. This directly follows from ${\bf P1}$ and the fact that $w^*(\set{X,Y}) = 1$.
        }
    }

    Now, we show how to construct $\nu^*(.)$. If $\nu'(.)$ is already half-integral, then we set $\nu^*(X) = \nu'(X)$ for all $X \in V_P$ and finish. Otherwise, set
    \myitems{
        \item $\nu^*(X) = \nu'(X)$ for every $X\in V_P$ with $\nu'(X) \in \set{0,1,\frac{1}{2}}$;
        \item $\nu^*(X) = 0$ for every $X\in V_P$ with $0 < \nu'(X) < 1/2$;
        \item $\nu^*(X) = 1$ for every $X\in V_P$ with $1/2 < \nu'(X) < 1$.
    }
    \vgap

    We need to verify that $\nu^*(.)$ meets the conditions listed in the statement of Lemma~\ref{lmm:vertex-packing}. Clearly, $\nu^*(X) = \nu'(X) = 1/2$ for every vertex $X$ in $\mycycle_1, ... \mycycle_\alpha$ (property {\bf P3}). For each $j \in [\beta]$, if $\mystar_j$ has at least two edges, then $\nu^*(X) =\nu'(X) = 0$ holds for the star's center $X$ and $\nu^*(Y) = \nu'(Y) = 1$ holds for every petal $Y$ (property {\bf P4-1}). If $\mystar_j$ has only one edge $\set{X,Y}$ --- that is, $\mystar_j$ is a one-edge star --- property {\bf P4-2} tells us that $\nu'(X) + \nu'(Y) = 1$. If $\nu'(X) = \nu'(Y) = 1/2$, then $\nu^*(X) + \nu^*(Y) = 1/2+1/2 = 1$. Otherwise, one vertex of $\set{X, Y}$ --- say $X$ --- satisfies $0 < \nu'(X) < 1/2$, and the other vertex satisfies $1/2 < \nu'(Y) < 1$. Our construction ensures that $\nu^*(X) + \nu^*(Y) = 0+1 = 1$.

    \vgap

    It remains to check that $\set{\nu^*(X) \mid X \in V_P}$ is indeed an optimal solution to the vertex-pack LP. First, we prove the the solution's feasibility. For this purpose, given any edge $F = \set{X, Y} \in E_P$, we must show $\nu^*(X) + \nu^*(Y) \le 1$. Assume there exists an edge $F = \set{X, Y} \in E_P$ such that $\nu^*(X) + \nu^*(Y) > 1$. Because $\nu^*(.)$ is half-integral, in this situation at least one vertex in $\set{X, Y}$ --- say $X$ --- must receive value $\nu^*(X) = 1$. We proceed differently depending on the value of $\nu^*(Y)$.
    \myitems{
        \item $\nu^*(Y) = 1/2$. By how $\nu^*$ is constructed from $\nu'$, in this case we must have $\nu'(X) > 1/2$ and $\nu'(Y) = 1/2$. But then $\nu'(X) + \nu'(Y) > 1$, violating the fact that $\nu'(.)$ is a feasible solution to the vertex-pack LP.

        \item $\nu^*(Y) = 1$. By how $\nu^*$ is constructed, we must have $\nu'(X) > 1/2$ and $\nu'(Y) > 1/2$. Again, $\nu'(X) + \nu'(Y) > 1$, violating the feasibility of $\nu'(.)$.
    }

    Finally, we will prove that $\set{\nu^*(X) \mid X \in V_P}$ achieves the optimal value for the vertex-pack LP's objective function. This is true because
    \myeqn{
        \sum_{X\in V_P} \nu^*(X)
        &=&
        \Big(
        \sum_{X: \text{$X$ not in any one-edge star}} \nu^*(X) \Big)
        +
        \sum_{X: \text{$X$ in a one-edge star}} \nu^*(X) \nn \\
        &=&
        \Big(\sum_{X: \text{$X$ not in any one-edge star}} \nu'(X) \Big)
        +
        \text{number of one-edge stars} \nn \\
        \explain{property {\bf P4-2}}
        &=&
        \Big(\sum_{X: \text{$X$ not in any one-edge star}} \nu'(X) \Big)
        +
        \sum_{X: \text{$X$ in a one-edge star}} \nu'(X) \nn \\
        &=&
        \sum_{X\in V_P} \nu'(X) \nn
    }
    and $\nu'(.)$ is an optimal solution to the vertex-pack LP.
\end{proof}

We are now ready to prove Lemma~\ref{lmm:app:divide_vertices}. Let $\set{\nu^*(X)\mid X\in V_P}$ be an optimal solution to the vertex-pack LP problem promised by Lemma~\ref{lmm:vertex-packing}. Divide $V_P$ into three subsets
\myeqn{
    U_A &=& \set{X \in V_P \mid \nu^*(X) = 1} \nn \\
    U_B &=& \set{X \in V_P \mid \nu^*(X) = 0} \nn \\
    U_C &=& \set{X \in V_P \mid \nu^*(X) = 1/2}. \nn
}

By the feasibility of $\nu^*(.)$ to the vertex-pack LP, no vertex in $U_A$ can be adjacent to any vertex in $U_C$, and no two vertices in $U_A$ can be adjacent to each other. It remains to verify that the sizes of $U_A$, $U_B$, and $U_C$ meet the requirement in the first bullet of Lemma~\ref{lmm:app:divide_vertices}. To facilitate our subsequent argument, given a one-edge star $\mystar_j$ (for some $j \in [\beta]$), we call it a {\em half-half one-edge star} if $\nu^*(X) = \nu^*(Y) = 1/2$, where $\set{X, Y}$ is the (only) edge in $\mystar_j$. Define
\myeqn{
    s &=& \text{number of half-half one-edge stars in $\set{\mystar_1, ..., \mystar_\beta}$}. \nn
}
On the other hand, if a one-edge star $\mystar_j$ is not half-half, we call it a {\em 0-1 one-edge star}.

\vgap

By Lemma~\ref{lmm:vertex-packing}, $U_C$ includes all the vertices in $\mycycle_1, ..., \mycycle_\alpha$ and all the vertices in the half-half one-edge stars. Hence, $|U_C| = k_\mit{cycle}+ 2s$. On the other hand, $U_B$ includes (i) the center of every star that has at least two edges and (ii) exactly one vertex from every 0-1 one-edge star. In other words, every star contributes 1 to the size of $U_B$ except for the half-half one-edge stars, implying $|U_B| = \beta-s$. Finally, $|U_A| = k-|U_B|-|U_C| = k_\mit{star}-\beta-s$. This completes the proof of Lemma~\ref{lmm:app:divide_vertices}.

\section{Analysis of the Sampling Algorithm in Section~\ref{sec:undirected-subgraph}} \label{app:undirected-sampling}

The purpose of preprocessing is essentially reorganizing the data graph $G = (V, E)$ in the adjacency-list format, which can be easily done in $O(|E|)$ expected time. The following discussion focuses on the cost of extracting a sample. We will first consider $\lambda \le \sqrt{|E|}$ before attending to $\lambda > \sqrt{|E|}$.

\extraspacing {\bf Case $\bm{\lambda \le \sqrt{|E|}}$.} Let $G_\sub = (V_\sub, E_\sub)$ be an occurrence of $P$ in $G$. There exist a constant number $c$ of isomorphism bijections $g: V_P \rightarrow V_{\sub}$ between $P$ and $G_\sub$ (the number $c$ is the number of automorphisms of $P$). Fix any such bijection $g$. Each repeat of our algorithm builds a map $f: V_p \rightarrow V$. It is rudimentary to verify that $\Pr[f = g] = \fr{1}{2|E|} \cdot \fr{1}{\lambda^{k-2}}$. Hence, each occurrence is returned with probability $\fr{c}{2|E|} \cdot \fr{1}{\lambda^{k-2}}$. It is now straightforward to prove that our algorithm has expected sample time: $O(|E| \cdot \lambda^{k-2} / \max\{1, \Out\})$.

\extraspacing {\bf Case $\bm{\lambda > \sqrt{|E|}}$.} If $\Out = 0$, the two-thread approach allows our algorithm to terminate in $O(\polym_\undir(|E|, \lambda, P)) = O(m^{\frac{k_\mit{cycle}}{2}+\beta} \cdot \lambda^{k_\mit{star} - 2\beta})$ time. Next, we consider only $\Out \ge 1$.

\vgap

For each $i \in [\alpha]$, denote by $k(\mycycle_i)$ the number of vertices in $\mycycle_i$; for each $j \in [\beta]$, denote by $k(\mystar_j)$ the number of vertices in $\mystar_j$. The fractional edge cover number of $\mycycle_i$ ($i \in [\alpha]$) is $\rho^*(\mycycle_i) = k(\mycycle_i)/2$. This means $\sum_{i=1}^k \rho^*(\mycycle_i) = k(\mycycle_i)/2 = k_\mit{cycle}/2$.

\vgap

Fix an arbitrary occurrence $G_\sub = (V_\sub, E_\sub)$  of $P$ in $G$. Let $g$ be any of the $c$ isomorphism bijections  between $P$ and $G_\sub$. Each repeat of our algorithm constructs a map $f: V_P \rightarrow V$ through isomorphism bijections $f_{\mycycle_1}, ..., f_{\mycycle_\alpha}, f_{\mystar_1}, ..., f_{\mystar_\beta}$. The event $g = f$ happens if and only if all of the following events occur.
\myitems{
    \item Event {\bf E$_{\mycycle_i}$} ($i \in [\alpha]$): $g(A) = f_{\mycycle_i}(A)$ for each vertex $A$ of $\mycycle_i$;
    \item Event {\bf E$_{\mystar_j}$} ($j \in [\beta]$): $g(A) = f_{\mystar_i}(A)$ for each vertex $A$ of $\mystar_j$.
}
If $c_{\mycycle_i}$ ($i \in [\alpha]$) is the number of automorphisms of $\mycycle_i$, we have
\myeqn{
    \Pr[\text{{\bf E$_{\mycycle_i}$}}]
    &=&
    \fr{1}{c_{\mycycle_i} \cdot |\occ(G, \mycycle_i)|}. \label{eqn:app:undirected-sample:help1}
}
Likewise, if $c_{\mystar_j}$ ($j \in [\beta]$) is the number of automorphisms of $\mystar_j$, we have
\myeqn{
    \Pr[\text{{\bf E$_{\mystar_j}$}}]
    &=&
    \fr{1}{c_{\mystar_j} \cdot |\occ(G, \mystar_j)|}. \label{eqn:app:undirected-sample:help2}
}
It follows from \eqref{eqn:app:undirected-sample:help1} and \eqref{eqn:app:undirected-sample:help2} that
\myeqn{
    \Pr[g = f]
    &=&
    \Big( \prod_{i=1}^\alpha \fr{1}{c_{\mycycle_i} \cdot |\occ(G, \mycycle_i)|} \Big) \cdot
    \Big( \prod_{j=1}^\beta \fr{1}{c_{\mystar_j} \cdot |\occ(G, \mystar_j)|} \Big).
    \nn
}
Therefore:
\myeqn{
    \Pr[\text{$G_\sub$ sampled}]
    &=&
    c
    \cdot
    \Big( \prod_{i=1}^\alpha \fr{1}{c_{\mycycle_i} \cdot |\occ(G, \mycycle_i)|} \Big) \cdot
    \Big( \prod_{j=1}^\beta \fr{1}{c_{\mystar_j} \cdot |\occ(G, \mystar_j)|} \Big). \label{eqn:app:undirected-sample:help3}
}
As this probability is identical for all $G_\sub$, we know that each occurrence of $P$ is sampled with the same probability.

\vgap

The expected number of repeats to obtain a sample occurrence of $P$ is
\myeqn{
    O\Big(
        \fr{
                \prod_{i=1}^\alpha  |\occ(G, \mycycle_i)|
            \cdot
                \prod_{j=1}^\beta  |\occ(G, \mystar_j)|
        }{\Out}
    \Big) \nn
}
whereas each repeat runs in time at the order of
\myeqn{
    \Big(\prod_{i=1}^\alpha \fr{|E|^{\rho^*(\mycycle_i)}}{|\occ(G, \mycycle_i)|}\Big)
    \cdot
    \Big(\prod_{j=1}^\beta \fr{|E| \cdot \lambda^{k(\mystar_j) - 2}}{|\occ(G, \mystar_j)|}\Big).
}
We can now conclude that the expected sample time is at the order of
\myeqn{
    \fr{1}{\Out}
    \cdot
    \Big(\prod_{i=1}^\alpha |E|^{\rho^*(\mycycle_i)} \Big)
    \cdot
    \Big(\prod_{j=1}^\beta |E| \cdot \lambda^{k(\mystar_j) - 2} \Big)
    &=&
    \fr{|E|^{k_\mit{cycle}/2} \cdot \lambda^{k_\mit{star} - 2}}{\Out}. \nn
}

\bibliographystyle{plainurl}
\bibliography{ref}

\end{document}